\documentclass[9pt, journal]{IEEEtran}
\usepackage{ifthen}
\usepackage{amsmath,bm}
\usepackage{soul}
\IEEEoverridecommandlockouts                              
\usepackage[dvipsnames]{xcolor}
\usepackage{url}
\usepackage{csquotes}

\usepackage{amsfonts,amssymb} 
\usepackage{mathtools,lipsum, nccmath}
\usepackage{graphics} 
\usepackage{mathptmx} 
\usepackage{times} 
\usepackage{cancel}
\usepackage{mathtools}

\usepackage{algorithm}
\usepackage{algorithmic}
\usepackage{lingmacros}
\usepackage{color}
\usepackage{float}
\usepackage{hyperref}
\usepackage{cleveref}
\usepackage{bm}
\usepackage{array}
\usepackage{tabu}

\hypersetup{
	colorlinks=true,
	linkcolor=blue,
	filecolor=blue,      
	urlcolor=blue,
	citecolor=blue
}

\newcommand{\kl}{Kullback-Leibler }

\usepackage{amsthm}

{
	\newtheorem{assumption}{Assumption}[section]

	\newtheorem{theorem}{Theorem}[section]
	
	\newtheorem{lemma}{Lemma}[section]
	\newtheorem{definition}{Definition}
	\newtheorem{proposition}{Proposition}[section]%
	\newtheorem{remark}{Remark}[section]
	
}
\let\oldproofname=\proofname
\renewcommand{\proofname}{\rm\bf{\oldproofname}}
\makeatletter
\renewenvironment{proof}[1][\proofname]{%
	\par\pushQED{\qed}\normalfont%
	\topsep6\p@\@plus6\p@\relax
	\trivlist\item[\hskip\labelsep\bfseries#1\@addpunct{.}]%
	\ignorespaces
}{%
	\popQED\endtrivlist\@endpefalse
}
\makeatother

\DeclareMathOperator{\Tr}{Tr}
\DeclareMathOperator{\col}{col}

\newcommand{\mbs}{\mathbf{s}}
\newcommand{\mba}{\mathbf{a}}
\newcommand{\mbf}{\mathbf{f}}
\newcommand{\mbe}{\mathbf{e}}
\newcommand{\mbw}{\mathbf{w}}
\newcommand{\mbF}{\mathbf{F}}

\newcommand{\mbSigma}{\pmb{\Sigma}}
\newcommand{\mbLambda}{\pmb{\Lambda}}
\newcommand{\mbsigma}{\pmb{\sigma}}
\newcommand{\mbmu}{\pmb{\mu}}
\newcommand{\mbepsilon}{\pmb{\epsilon}}
\newcommand{\mbzeta}{\pmb{\zeta}}
\newcommand{\mbz}{\mathbf{z}}

\newcommand{\bluecolor}{\color{blue}}

\newcommand{\mbbP}{\mathbb{P}}
\newcommand{\mbbQ}{\mathbb{Q}}

\newcommand{\figsmallwidth}{0.330\linewidth}

\usepackage{soul}
\let\oldmb\mathbold
\protected\def\mathbold{\oldmb}

\usepackage{graphicx}
\usepackage{verbatim}   
\usepackage{color}
\usepackage{cancel}
\usepackage[normalem]{ulem}
\usepackage{subcaption}
\newcommand*{\rom}[1]{\expandafter\@slowromancap\romannumeral #1@}
\makeatother

\begin{document}
	\newcounter{keep_kde_plots}
	\setcounter{keep_kde_plots}{1}
	%
	
	\title{Dynamic Programming Approximate Optimal Control for Model-Based Reinforcement Learning}
	\author{Prakash Mallick\textsuperscript{1},  Zhiyong Chen\textsuperscript{2}\\
		\textsuperscript{1}Australian Institute of Machine Learning (AIML), 
		\textsuperscript{2}University of Newcastle\\

			\thanks{Prakash Mallick is with the Australian Institute of Machine Learning (AIML) affiliated to University of Adelaide and Zhiyong Chen is with the School of Engineering,
				University of Newcastle, Callaghan, NSW 2308, Australia and  Emails: {\tt\small prakash.mallick@adelaide.edu.au, zhiyong.chen@newcastle.edu.au}}%
		}

	\markboth{}%
	{Schaich \MakeLowercase{\textit{et al.}}: \thetitle}
	%

	
	\maketitle
	\begin{abstract}                          
		
This article proposes an improved trajectory optimization approach for stochastic optimal control of dynamical systems affected by measurement noise by combining optimal control with maximum likelihood techniques to improve the reduction of the cumulative cost-to-go. A modified optimization objective function that incorporates dynamic programming-based controller design is presented to handle the noise in the system and sensors. Empirical results demonstrate the effectiveness of the approach in reducing stochasticity and allowing for an intermediate step to switch optimization that can allow an efficient balance of exploration and exploitation mechanism for complex tasks by constraining policy parameters to parameters obtained as a result of this improved optimization. This research study also includes theoretical work on the uniqueness of control parameter estimates and also leverages a structure of the likelihood function which has an established theoretical guarantees. Furthermore, a theoretical result is also explored that bridge the gap between the proposed optimization objective function and existing information theory (relative entropy) and optimal control dualities. 
	\end{abstract}
	
	\begin{IEEEkeywords}
		Stochastic systems, optimal control, trajectory optimization, {robust control},
		maximum likelihood, {robust} expectation maximization 
	\end{IEEEkeywords}

	%
	\IEEEpeerreviewmaketitle
	\section{Introduction}
	
	Model-based trajectory optimization offers a general method of planning, and therefore has been utilized to generate/reproduce a wide range of behaviours; see e.g., \cite{todorov2005generalized,toussaint2009robot,Schulman13findinglocally}. However, such methods does not directly handle the partially observed environments. Research study conducted by~\cite{mallick2022stochastic} and~\cite{mallick2022reinforcement} handled the latency in states fairly well over a 2D task; however these proposed paradigms have not been tested on more complicated tasks. This paper delves into empirical study of the expectation-maximization-based stochastic control approach widely evaluated on complicated robotic simulation framework assigned to carry out complex tasks. The new modified optimization approach leverages the fundamentals of dynamic programming and the newly formulated optimization routine has the same functional structure that of the likelihood. As a consequence of this, all of the theoretical results presented in the research studies of~\cite{mallick2022stochastic} can be leveraged into this newly proposed optimization procedure. The modified objective function is experimented on two complex tasks involving a \texttt{PR2} similar robotic simulation platform i.e., \texttt{MUJOCO}~\cite{mujoco}; details of which will be described later in the chapter.

{
The contributions of this research can be summarized as follows, with references to relevant works:
\begin{itemize}
    \item The research methods developed in \cite{mallick2022stochastic,mallick2022reinforcement} propose a sophisticated approach to address trajectory optimization, particularly in the presence of measurement noise. While Expectation-Maximization (EM) methods are effective in handling measurement noise, the simplicity and computational efficiency of iLQG \cite{todorov2005generalized} make it highly appealing for practical and industrial applications. Therefore, the subsequent sections of this chapter aim to integrate iLQG and EM-based stochastic optimal control (SOC) into a unified framework, which is evaluated on complex tasks.
    \item In order to introduce a novel dynamic programming-based optimization for locally optimal controller design, the optimization objective function has been modified. Building on the approach in \cite{chebotar2017combining}, the locally optimal solutions of the newly proposed objective function guide the iLQG control design routine. The new policy is constrained by previous policies, and significant results are presented in the simulation experiments.
    \item 
Moreover, inspired by the research connection between information theory and optimal control dualities, particularly the work in \cite{theodorouinformation}, this study leverages these connections, with a focus on free energy \cite{theodorou2015nonlinear}. The developed model can be interpreted as a special case of information theory (relative entropy) and optimal control dualities.
\item 
Theoretical foundations establish the uniqueness characteristic of the maximizer of the likelihood function based on dynamic programming. This provides a practical and viable approach to achieve a lower-dimensional approximation in each iteration. Additionally, it is demonstrated that this new optimization technique effectively explores the state space, which is characterized by high uncertainty in the control actions of the robotic joints of the manipulator. Furthermore, numerical evidence showcases improved state trajectories with reduced stochasticity and better reduction in the cumulative cost-to-go.
\end{itemize}
}

 	The paper is organised in the following manner: Section~\ref{chap5:problemstatementandprocedure} elucidates fundamental definitions of \kl divergence, relative entropy and free energy in a manner pertinent to the context of optimal control. Then, section~\ref{chap5:problemstatementandprocedure} delineates the underlying mathematical problem statement, provides an explanation of the procedure of the solution. This is followed by Section~\ref{sec:DPSOC}, which demonstrates the new DP optimization and a few theoretical results which is new in accordance to the new objective function. Then, Section~\ref{sec:connection} of this manuscript focusses on providing an information theoretic view of the EM-based stochastic optimal control approach.
	
	\begin{table}[t]
		\caption{Summary of symbols}
		\centering
		\begin{tabular}{ |p{2.5cm}|p{5.5cm} | } 
			\hline
			\textbf{Symbol}& \textbf{Definition} \\
			$k,T$ & Time instant, Length of episode \\
			$\mathbf{s}_k $   &  Measured state  { (implementation) } \\
			& or { latent state (optimization) at time instant $k$ }\\
			$\mathbf{x}_k ,\mathbf{a}_k$   & Real state  and control action  at time instant $k $ \\
			$Y_k(\mathbf{s}_k,\mathbf{a}_k)$   or 
			$Y_k(\mathbf{s}_k,\phi_k)$   
			& Instantaneous cost at time instant $k$ \\
			$y_k$ & Observed cost  $p(Y_k)$ \\
			{${\mathbb{S}_{T+1}}$} &  {Measured or latent variable $\{\mathbf{s}_1,\mathbf{s}_2,\cdots, \mathbf{s}_{T+1}\}$}\\
			${\mathbb{Y}_{T}}$/${\mathbb{A}_{T}}$   & Reward observation   $\{y_1, y_2, y_3,\cdots, y_{T}\}$/Set of actions $\{\mba_1, \mba_2, \mba_3,\cdots, \mba_{T}\}$\\
			$\hat{\phi}^{i}$ & Controller parameter $\phi$ at the  $i$-th iteration\\
			$\mathbb{E}$ & Expectation of a random variable \\ 
			$V_\phi ( {\mathbb{S}_{T+1}})$ & Cumulative sum of expected costs \\
			$L_\phi ({\mathbb{Y}_T})$ , 	$\mathcal{L}(\phi,\hat{\phi}^i)$ & Observation log-likelihood, Mixture likelihood \\
			$\text{vec}(\cdot )$ & Column vector stacked by vector  of its matrix argument \\
			$\col (\cdots)$ & Column vector stacked by its column arguments \\
			$\Tr (\cdot)$ & Trace of its matrix argument\\
      $(\Omega,\mathcal{F})$  & Sample space of an experiment, $\sigma-$algebra respectively. \\
     $\pmb{{P}} (\Omega)$ &   probability measure on $\sigma-$algebra $\mathcal{F}$\\
			$\nabla$, $\nabla^2$  &  Gradient vector field, Hessian matrix; second-order partial derivative of a scalar function \\
			$^\top$, 	$\otimes$ & Transpose operator, Kronecker product operator \\
			$\mathbb{R}$ /  $\mathbb{R}^+$ & Set of real numbers / positive numbers \\
			$ \mathbf{I}_{(s)} $ & Identity matrix (of dimension $s$)  \\
			\hline
		\end{tabular}\label{summary}
	\end{table}

	\section{Preliminaries, Problem Statement and Solution Procedure} \label{chap5:problemstatementandprocedure}
	In this section, the dynamic model being investigated, the problem formulation, and the proposed solvability procedure are presented. Additionally, the mathematical notations involved in addressing the problem discussed in this paper are elaborated upon. For a summary of the symbols used, readers may refer to Table~\ref{summary}. 
	
	\subsection{Preliminaries}  \label{sec:prelim}
 In this section, we briefly review the key concepts underlying the fundamental duality relationships between free energy and relative entropy~\cite{dai1996connections}. Definition~\ref{def:FreeEnergy} presents the concept of \enquote{free energy}, involving measurable function and a probability measure on a measurable space. The free energy is defined as the expectation of the function under the probability measure, and it is obtained by taking the logarithm of the integral of the exponential of the function with respect to the probability measure. The function captures a state-dependent cost and includes a terminal cost as well as a sum of costs over a sequence of states. The parameter $\rho$ determines a measure of risk associated with the free energy. Further details on this definition can be found in~\cite[Section 3]{theodorou2015nonlinear}.
	
	Definition~\ref{def:Entropy} introduces the notion of \enquote{relative entropy} or \enquote{Kullback-Leibler divergence}. It measures the difference between two probability measures, denoted as $\mathbb{P}$ and $\mathbb{Q}$. The relative entropy of $\mathbb{Q}$ with respect to $\mathbb{P}$ is defined as an integral involving the logarithm of the ratio of the densities of $\mathbb{Q}$ and $\mathbb{P}$. Further details and formal expressions can be found in~\cite{dai1996connections}
	
	\begin{definition}  \label{def:FreeEnergy}
		Let the pair $(\Omega,\mathcal{F})$ denote a measurable space\footnote{The definition of measurable space can be found~\cite[Page 160]{Billingsley1995probability}}, where $\Omega$ and $\mathcal{F}$ corresponds to symbols mentioned in~\ref{summary}, and let $\pmb{{P}} (\Omega)$ define a probability measure\footnote{Please refer to~\cite[Page 160]{Billingsley1995probability}} on $\sigma-$algebra $\mathcal{F}$. Consider $\mathbb{P} \in \pmb{{P}}( \Omega )$ and the function $\mathcal{J}(\cdot): \mathcal{F}  \rightarrow \mathbb{R}$ be a measurable function. Then the expectation term:
		\begin{align}
			\mathbb{E} \Big( \mathcal{J} (\mbs_k,k) \Big) = \log \int  e^{(\rho \mathcal{J} (\mbs_k,k))} d\mathbb{P}
		\end{align}
		is called \textbf{free energy} of $\mathcal{J}$ with respect to $\mathbb{P}$, where $\mathcal{J}$ can be defined as, 
		\begin{align}
			\mathcal{J}(\cdot) \triangleq \mathcal{J} (\mbs_k,k) = g_{T} (\mbs_T)  +   \sum_{j=k}^{T-1} g_j(\mbs_j,\mbw_j) 
		\end{align}
		a state dependent cost as described in~\cite[Section 3]{theodorou2015nonlinear} and $\mathbf{w}_k$ is a r.v. which corresponds to noise. The term $\rho$ is a parameter which determines some measure of risk\textemdash please see~\cite{theodorou2015nonlinear,theodorouinformation} for details.
	\end{definition}
	\begin{definition}  \label{def:Entropy}
		Let $\mathbb{P}\in \pmb{{P}} ({ \Omega	 })$, the \textbf{relative entropy} of $\mathbb{P}$ with respect to $\mathbb{Q}$ be defined as:
		\begin{equation}
		D_{KL} (\mathbb{Q} || \mathbb{P}  )=\begin{cases}
				\int \log \frac{d\mathbb{Q}}{d\mathbb{P}} \mathbb{Q} , & \text{if $ \mathbb{Q} <<\mathbb{P}$ and $ \log \frac{d\mathbb{Q}}{d\mathbb{P}} d\mathbb{Q}  \in L^1$ }.\\
				+\infty, & \text{otherwise}.
			\end{cases}
		\end{equation}
		The term $L^1$ mentioned above describes functions on $[0,\infty)$ integrable in the Lebesgue sense, and where \enquote{$<<$} denotes the absolute continuity of $\mathbb{Q}$ with respect to $\mathbb{P}$. If one says that $\mathbb{Q}$ is absolutely continuous with respect to $\mathbb{P}$, one can write $\mathbb{Q}<<\mathbb{P}$  if $\mathbb{P}(H)=0 \implies \mathbb{Q}(H)=0, \forall H \in \mathcal{F}$, where $\mathcal{F}$  denotes a $\sigma-$algebra.
	\end{definition}
	After having described few definitions, the next sub-section focuses on mathematical modelling of the dynamical systems, controller parameter, the objective function and a brief step-by-step procedure to attain the objective.

\subsection{Mathematical notation and modeling}

The paper takes into account a stochastic dynamics that does not have a known model
from first principles,  in the presence of uncertainties such as parameter variation, external disturbance, 
sensor noise, etc.
The completed system is considered to be a global model, $O$, that is composed of multiple local models 
$o^l,\;  l = \{1, 2, \cdots\}$, and each of which follows an MDP, 
called a local model. 
We are   interested in a finite-horizon optimal control for a particular initial state, rather than for all possible initial states. {The procedure to handle numerous initial states is laid out in \cite{mallick2022reinforcement} which utilizes the policy structure of this paper to generalize to numerous initial conditions.
}

The {partially observable MDP} encompasses a {\it latent state} $\mathbf{s}_{k} \in \mathbb{R}^{n_s}$ 
and a {\it control action} $\mathbf{a}_k \in \mathbb{R}^{n_a}$, at  time instant $k = 1, 2, \cdots$, and the 
local state transition dynamic model is represented by a conditional probability density function (p.d.f.) expressed as: 
\begin{align} \label{modelpdf}
	p (\mathbf{s}_{k+1} | \mathbf{s}_k, \mathbf{a}_k).
\end{align}
In particular, for $k=1$,  $\mathbf{s}_1 \in \mathbb{R}^{n_s}$ is called the {\it initial state}, obeying a specified distribution. Variables $n_s$ and $n_a$ are integers which are the dimensions of state
and action space.
We specifically consider a finite-horizon MDP in this paper  for $k = 1, 2, \cdots, T$, called an 
{\it episode}, with the time instant $T$ being the end of episode. It is  worth mentioning that the 
p.d.f. in \eqref{modelpdf} varies with time $k$ and 
the time-varying nature is capable of characterizing more complicated dynamical behaviors but also brings more challenges 
in control design. It will be elaborated in Section~\ref{sec:DPSOC}. 

The entity $Y_k(\mathbf{s}_k,\mathbf{a}_k) \in \mathbb{R}^+$ denotes the instantaneous real valued {\it cost} for executing action $\mathbf{a}_k$ at state $\mathbf{s}_k$. 
It has a more specific expression as follows,
\begin{align} \label{quad_reward}
	Y_k(\mathbf{s}_k,\mathbf{a}_k) = (\mathbf{s}_k-\mathbf{s}^*)^\top \mathbf{Q_s} (\mathbf{s}_k-\mathbf{s}^*) + (\mathbf{a}_k-\mathbf{a}^*)^\top \mathbf{Q_a} (\mathbf{a}_k-\mathbf{a}^*),
\end{align} where $\mathbf{s}^*$ and $\mathbf{a}^*$ are the target state and 
control action, respectively, and  $\mathbf{Q_s} > 0$ and $\mathbf{Q_a}> 0$ are some specified matrices. 
As $\mathbf{s}_k$ and $\mathbf{a}_k$ are random variables,  
$ Y_k(\mathbf{s}_k,\mathbf{a}_k)$  (with $Y_k$ a continuous and deterministic function) is also a random variable, shorted 
as $Y_k$. We develop another {variable, i.e., $y_k=F(Y_k) \in \mathbb{R}^+$ (known as \textit{observed cost}) which is the exponential transformation of the immediate}  cost $Y_k(\mathbf{s}_k,\mathbf{a}_k)$ following a p.d.f. 
$ p (y_k | \mathbf{s}_k, \mathbf{a}_k)$, which  will be elaborated later.

Overall, the MPD consists of the transition dynamics $p(\mathbf{s}_{k+1}|\mathbf{s}_k,\mathbf{a}_k)$ and the cost observation p.d.f. $p(y_k|\mathbf{s}_k,\mathbf{a}_k)$ in an augmented form, i.e.,{
	\begin{align} \label{ltv_eq}
		{ p \Big(  \begin{bmatrix}
				\mathbf{s}_{k+1}       \\
				y_k      
			\end{bmatrix} | \mathbf{s}_k, \mathbf{a}_k \Big) }   = \mathcal{N} \Big( { \mathbf{A}^o_k }
		{ \begin{bmatrix}
				\mathbf{s}_k       \\
				\mathbf{a}_k       
		\end{bmatrix} }  ,{ \pmb{\Sigma}^o_k  \Big)},
	\end{align} 
	which is referred to as the {\it dynamic model} in the subsequent parts of the paper.}
{The symbol $\mathcal{N} (\cdot,\cdot)$ is used to denote a Gaussian distribution with the two parameters 
	of mean and variance.}
A time-varying linear Gaussian p.d.f. is used in  \eqref{ltv_eq} as 
an approximation of a real model which is in general nonlinear,  where 
the matrices $ \mathbf{A}^o_k$ and $\pmb{\Sigma}^o_k$ is  to be determined in Section~\ref{learning_section}
using the dynamic model fitting technique. 
\begin{figure}[t]
	\centering
	\includegraphics[scale=0.335]{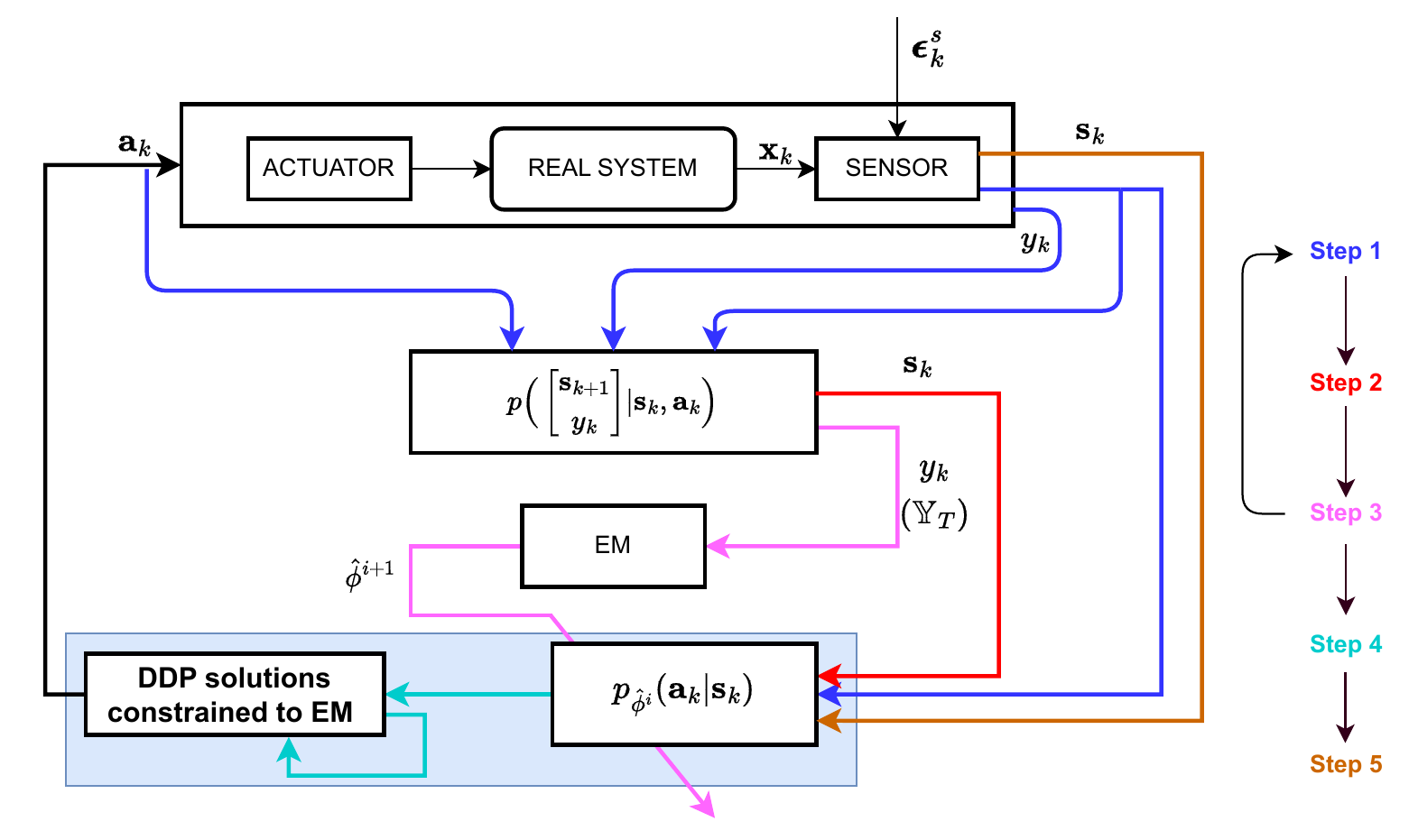}
	\caption{Schematic diagram of the dp-based design procedure in five steps.}
	\label{fig:overall}
\end{figure}
\subsection{Controller parameter space} 
This subsection presents the definition of parameter space of a controller that is utilized in the paper. The control action is sampled from a  linear Gaussian p.d.f. that describes the policy as shown below,
\begin{align}  
	p_{{\phi_k}}({\mathbf{a}}_k|\mathbf{s}_k) = \mathcal{N} ( \mbF_k \mathbf{s}_k + \mbe_k ,{\mbSigma_k} )  ,  \label{controlChap1}
\end{align}
for some matrices $\mbF_k, \mbSigma_k$ and a vector $\mbe_k$, representing state feedback control. {\bluecolor }
The matrix $\mbSigma_k$ is symmetric positive definite, and
$\mbSigma_k^{\frac{1}{2}}$ is the square root of ${\mbSigma_k}$ satisfying 
${\mbSigma_k}  =(\mbSigma_k^{\frac{1}{2}}) ^\top \mbSigma_k^{\frac{1}{2}} $. 
Let $\mbf_k =  \text{vec} (\mbF_k)$ and
$\mbsigma_k =  \text{vec} (\mbSigma_k^{\frac{1}{2}} )$.
Then,  the vector 
\begin{align*} 
	{\phi}_k =\col (\mbf_k, \mbe_k, \mbsigma_k),
\end{align*} is called the  controller parameter vector. 
Over the episode under consideration, the controller parameters are lumped as follows, 
\begin{align}
	\phi  =\col (\phi_1, \phi_2, \cdots, \phi_T) \in \Phi,
\end{align}
{where ${{\Phi}}$ is a non-empty convex compact subset of $\mathbb{R}^{({n_a n_sT+n_aT} + {n_a n_aT} ) }$.
	The time-varying feature of the controller is represented by the variation of $\phi_k$ with $k$,
	which aims to account for the complexity of the dynamical system.
	A stochastic policy adopted in this paper lays down the exploration mechanism  in a reinforcement learning setting. 
}

{Before diving into our problem formulation this section provides a definition of the so-called  SOC problem (see e.g., \cite{stengel1994optimal}).} 
Specifically, 
for the stochastic dynamical model \eqref{ltv_eq},  the SOC problem is formulated as follows,  
\begin{align} \label{cumsum_obj}
	& \min_{\mathbf{a}_1, \mathbf{a}_{2}, . . ,\mathbf{a}_T} \mathbb{E} \sum_{k=1}^T Y_k(\mathbf{s}_k,\mathbf{a}_k),
\end{align}
where $\mathbf{s}_k$ and $\mathbf{a}_k$ are the
variables of the dynamical model p.d.f.  \eqref{ltv_eq}
and  the control action p.d.f. \eqref{controlChap1}.
{The expectation in the above equation is taken over next states which are a result of instantiations of the noise in the dynamical equation at a start state of $\mbs_1$.} To express the cost penalty to be explicitly dependent on 
$\phi_k$, we rewrite 
$Y_k(\mathbf{s}_k,\mathbf{a}_k)$ as $Y_k(\mathbf{s}_k,\phi_k)$ with slight abuse of notation. 
Also, we can rewrite \eqref{cumsum_obj} in terms of the control parameter $\phi$, i.e., 
\begin{align} \label{cumsum_obj1}
	&  \min_{\phi}  \mathbb{E}  V_\phi ( {\mathbb{S}_{T+1}})  \;
	\text{for} \;
	V_\phi ( {\mathbb{S}_{T+1}}) \triangleq   \sum_{k=1}^T Y_k(\mathbf{s}_k,\phi_k).
\end{align}
with ${\mathbb{S}_{T+1}}  =   \{\mathbf{s}_1,\mathbf{s}_2,\cdots, \mathbf{s}_{T+1}\}$.

{Please note that, exactly solving the optimization problem of Eq. \eqref{cumsum_obj1} will result in an optimal (global) control law $p_{\phi^*}(a_k)$ which is independent of the starting initial state. Nevertheless, exactly finding this control law is highly unlikely in POMDPs because it is intractable to solve \cite{10.2307/3689975} \cite{platt2010belief}. An optimal control problem for POMDP cannot be solved efficiently without approximations. Therefore, to this extent numerous approximate control strategies have been proposed in the literature of optimal control; see e.g.,\cite{li2004iterative}, \cite{tassa2012synthesis}, \cite{levine2014motor}, \cite{ziebart2010modeling(a)}, \cite{kappen2012optimal}, \cite{todorov2007linearly}. Additionally, in a POMDP version of reinforcement learning setting it has been proved that the solution is NP-complete; see e.g., \cite{littman1994memoryless} and \cite{lusena2001nonapproximability}. As a result, this attribute calls for approximations to POMDP solutions which makes it analytically tractable.}

Despite employing a comparable formulation, the approach to addressing the problem differs from that mentioned in~\cite[Section 4]{mallick2022stochastic}. A five-step methodology for addressing the problem statement is delineated below with reference to a schematic diagram as per~\ref{fig:overall}. The efficacy of this five-step process has been verified through empirical testing on a simulation framework, which will be expounded upon in the subsequent sections.

{\it \textbf{Step 1:} Dynamic model fitting:}
From an initial state $\mathbf{s}_1$ sampled from a specified distribution, 
the real system is operated with the controller \eqref{controlChap1}
for a pre-selected controller parameter $\phi = \hat\phi^0=\col (\hat\phi^0_1,\hat \phi^0_2, \cdots, \hat\phi^0_T)$
and the control actions $\{\mathbf{a}_1,\mathbf{a}_2,\cdots, \mathbf{a}_{T}\}$ 
and the states
$\{\mathbf{s}_1,\mathbf{s}_2,\cdots, \mathbf{s}_{T+1}\}$ are recorded.
Calculate $Y_k(\mathbf{s}_k,\hat\phi^0_k)$ and hence $y_k = e^{ - Y_k }$. 
{Then,  the dynamic model \eqref{ltv_eq} is identified by fitting it to 
	to the collected tuples of data $\{\mathbf{s}_k,\;\mathbf{a}_k,\;\mathbf{s}_{k+1},\;y_k\}$, $k=1,\cdots, T$.}

{\it \textbf{Step 2:} Generation of cost observation:}
From an initial state $\mathbf{s}_1$ sampled from a specified distribution, 
the   cost observations ${\mathbb{Y}_{T}}= \{y_1, y_2, y_3,\cdots, y_{T}\}$ are generated using
the dynamic model \eqref{ltv_eq}  (obtained from Step 1)
and the controller  \eqref{controlChap1} with the controller parameter $\phi = \hat\phi^0$.  

{\it \textbf{Step 3:} Optimization of control action:}  Let  ${ p_{\phi} (\mathbb{S}_{T+1} | \mathbb{Y}_T)}$ be the probability of 
the {\it latent states}
${\mathbb{S}_{T+1}}  = \{\mathbf{s}_1,\mathbf{s}_2,\cdots, \mathbf{s}_{T+1}\}$
given the observation ${\mathbb{Y}_{T}}$ (obtained from Step 2), obeying the 
closed-loop system composed of the dynamic model \eqref{ltv_eq}  (obtained from Step 1)
and the controller  \eqref{controlChap1} with a 
controller parameter $\phi$. 
The optimization of a local control policy is formulated as follows
\begin{align} \label{valuefn_ltvChap5}
	\phi^{*} = \arg\min_{\phi}\mathbb{E}_{ p_{\hat\phi^0} (\mathbb{S}_{T+1} | \mathbb{Y}_T)} V_\phi ( {\mathbb{S}_{T+1}}).
\end{align}

A practical approach for solving the optimization problem \eqref{valuefn_ltvChap5} is to use the following strategy,
\begin{align} \label{valuefn_ltvi3}
	{ \hat \phi^{i*}}  = \arg\min_{\phi}\mathbb{E}_{ p_{\hat\phi^i} (\mathbb{S}_{T+1} | \mathbb{Y}_T)} V_\phi ( {\mathbb{S}_{T+1}}) ,
\end{align}
recursively with $\hat \phi^{i+1} = { \hat \phi^{i*}}$,  for $i=0,1, \cdots$.
It is expected that $ \hat \phi^{i} $ approaches $\phi^*$ as $i$ goes to $\infty$. 

{\it \textbf{Step 4:} Implementation of optimization inside the framework of iLQG:} 

After carrying out the maximization, I utilize the optimal parameters
to again collect samples, carry out the dynamics fitting of piecewise linear
models and then with the help of those dynamics carry out another degree of optimization according to eq~\eqref{eqn:constrainObjectiveChap5}. The term $p_{\hat{\phi}^{i+1}}$ in eq~\eqref{eqn:constrainObjectiveChap5} denotes the distribution of state marginals when excited with parameters $\phi^{i+1}$. The parameters $\hat{\phi}^{i+1}$ utilized in eq~\eqref{eqn:constrainObjectiveChap5} are generated as a result of~\eqref{valuefn_ltvi3}.
\begin{align}  \label{eqn:constrainObjectiveChap5}
	\phi = \text{min}_\phi \mathbb{E}_{p_\phi }  \big[ Q(\mbs_k,\mba_k)  \big]  \text{ s.t. } D_{KL} (p_\phi   || p_{\hat{\phi}^{i+1}} ) < \nu_k,
\end{align}
where $\nu_k$ is a parameter which gives the degree to which the probability distributions are close. 

{\it \textbf{Step 5:} Evaluation:}
Run the real system with the controller $\phi$ from~\eqref{eqn:constrainObjectiveChap5} and evaluate the performance.

Further in the sequel we describe a new way to approximate eq~\eqref{valuefn_ltvi3} which is eventually a major contribution of this chapter.

{ 
	\begin{remark} 
		In Steps 1 and 5, the real dynamical system is operated for data generation and performance evaluation, respectively. 
		The state $\mathbf{s}_k$ is physically measured, which represents the observed system state carrying measurement noise. 
		However, in  Steps 2 and 3, only theoretical computation is conducted without operating the real system, thus leveraging the latency nature of states.
		In this context,  $\mathbf{s}_k$ represents the explored state
		obeying a  joint probability  ${ p_{\phi} (\mathbb{S}_{T+1} | \mathbb{Y}_T)}$ and and contains noise. Thus, the actual or true system states remain unobserved, and the model described in Equation \eqref{ltv_eq} is treated as a partially observable Markov decision process (POMDP). In Step 4, a DDP-based optimal control design is employed while constraining the policy derived from Step 3.
	\end{remark}

	The subsequent sections are concerned with the optimization problem \eqref{valuefn_ltvi3},
	which is regarded as the approximation of the original optimization problem \eqref{cumsum_obj1}. 
	It is easy to see that \eqref{valuefn_ltvi3}
	is equivalent to   \begin{align}  \label{valuefn_ltv3}
		\min_{\mathbf{a}_1, \mathbf{a}_{2}, . . ,\mathbf{a}_T} \mathbb{E}_{p_{\hat\phi^i} (\mathbb{S}_{T+1}|\mathbb{Y}_T)}
		\sum_{k=1}^T Y_k(\mathbf{s}_k, \mathbf{a}_k)   .
	\end{align}
	It is expected that recursively solving the problem  \eqref{valuefn_ltvi3} or \eqref{valuefn_ltv3} 
	will approach a solution to \eqref{cumsum_obj1}.
	However, the global convergence is a great challenge and the effectiveness can only be 
	numerically/empirically verified.  The next section describes the new optimization approach.

{ 
	
	\section{Dynamics Fitting}  \label{learning_section}

In this section, the authors first provide a specific definition of $y_k$ and an expression of its density function. Additionally, they provide a detailed explanation of the procedure formulated in Step 1 for obtaining linear time-varying parameter estimates of the dynamic model represented by equation \eqref{ltv_eq}. This procedure combines the methodology employed in~\cite{levine2016end} with existing variational Bayesian (VB) strategies used for a finite mixture model, as outlined in~\cite{bishop2006pattern}.

A specific definition of the variable $y_k$ is presented as follows:
\begin{align} \label{pdf_of_y_k}
y_k(\mbs_k,\mba_k) = e^{- Y_k (\mbs_k,\mba_k)}.
\end{align}

The variable $y_k(\mbs_k,\mba_k)$, which takes values in the interval $(0,1]$, characterizes the likelihood of the state-action pair $(\mbs_k,\mba_k)$ being in proximity to the optimal trajectory. A lower cost value $Y_k (\mbs_k,\mba_k)$ corresponds to a higher value of $y_k(\mbs_k,\mba_k)$, indicating an increased likelihood of proximity to the optimal trajectory. The exponential transformation has been demonstrated to be successful in determining the probability of an optimal event occurring in optimal control, as evidenced by previous studies such as~\cite{cooper2013method}.

 The following is an assumption for the p.d.f. of $Y_k$.
 \begin{assumption} \label{assumption:pdf_Yk}
The p.d.f. of $Y_k$  follows an exponential distribution with parameter $\lambda$, i.e.,
\begin{align}\label{pYk}
 p(Y_k) = \lambda e^{-\lambda Y_k} \text{ where }  \lambda>1.
\end{align}
 \end{assumption}
The use of an exponential distribution is supported by relevant work. For instance,~\cite{dayan1997using} assumes rewards (negative costs) are sampled from an exponential distribution, while~\cite{norouzi2016reward} utilizes an exponentiated payoff distribution to establish a connection between maximum likelihood and an optimal control objective.

\begin{lemma} \label{lemma_prob2}
For $Y_k$ of  the p.d.f. \eqref{pYk},
the random variable $y_k$ in \eqref{pdf_of_y_k} has a p.d.f. of the form 
\begin{align} \label{pdf_y_k}
   p(y_k) = {\lambda} { y_k  }^{{\lambda} -1} . 
\end{align}
\end{lemma}

\begin{proof}
Please see~\cite{mallick2022stochastic}.
\end{proof}

	One can run one iteration of experiment and collect tuples of measured 
	$\{\mathbf{x}_k,\;\mathbf{u}_k,\;\mathbf{x}_{k+1},\;y_k\}$ for one episode $k=1,\cdots, T$. Practically, 
	one can repeat the experiments for $M$ times from the same initial conditions with a random seed value to gather sufficiently many samples, each of which is denoted by 
	\begin{align*}
		\mathcal{D}_k^m=\{\mathbf{x}_k,\;\mathbf{u}_k,\;\mathbf{x}_{k+1},\;y_k\}_{\text{$m$-th experiment}},
	\end{align*}
	for $m=1,\cdots,M$.  Let 
	$\mathcal{D}_k = \{\mathcal{D}_k^1,\cdots, \mathcal{D}_k^M \}$ 
	and $\mathcal{D} = \{\mathcal{D}_1,\cdots, \mathcal{D}_T \}$.
	Then, one can fit a Gaussian mixture model (GMM) to the data set $\mathcal{D}$. 
	In particular, the VB inference method is used to determine the parameters of the GMM, i.e., the means, covariances and weights of the Gaussians. 
	
	The GMM produced as a result of VB inference acts as a considerable global prior and it helps in bringing in information to construct a solitary normal-inverse Wishart (NIW) distribution. This NIW acts as a conjugate prior for a Gaussian distribution 
	\begin{align} \label{NIWp}
		p (\mathbf{x}_k,\mathbf{u}_k,\mathbf{x}_{k+1},y_k) = \mathcal{N} (\mbmu_k, \mbLambda_k).
	\end{align}
	Next, it will be elaborated that the NIW prior plays an essential role in attaining the parameters, i.e., the mean $\mbmu_k$ and the covariance ${\mbLambda}_k$.

	The procedure of fitting GMM to $\mathcal{D}$ involves constructing NIW distributions to act as prior for means and covariances of Gaussian distributions involved in mixture model. In addition to it, Dirichlet distributions are defined to be the prior on the weights of the Gaussian distributions which would explain the mixing proportions of Gaussians. {Then, iterative VB strategy is adopted} to increase the likelihood of joint variational distribution (see e.g., \cite{bishop2006pattern}-Section 10.2) for attaining the parameters of GMM. 
	The attained parameters of GMM are utilized to further obtain the parameters of the solitary NIW prior which acts a representative of the global GMM prior. The purpose of NIW prior is to garner the information contained in the global GMM prior. The mean of the solitary NIW prior is $\mbmu^0_k = \sum_{c=1}^{C} ( w^c_k  \mbmu^c_k )$ where $w^c_k$ and $\mbmu^c_k$ are the weight and mean of the $c$-th  Gaussian in the GMM and ${C}$ is the total number of initialized Gaussian clusters.  The precision matrix $\mbLambda^0_k$ of the solo NIW prior is evaluated by calculating the deviation of each cluster from $\mbmu^0_k$. There are two more essential parameters of the solo NIW conjugate prior namely $n_0$ and $k_0$, which are set for the total $M$ samples. 
	Define the empirical mean, $\mbmu^{emp}_{k}\in \mathbb{R}^{n_u+2n_x+1}$ and the covariance $\pmb{\Lambda}^{emp}_k \in \mathbb{R}^{{(n_u+2n_x+1)}\times {(n_u+2n_x+1)}}$ as follows, for the data set $\mathcal{D}_k$, 
	\begin{align}
		&  \mbmu^{emp}_{k}=\frac{1}{M}\sum_{m=1}^{M} {\mathcal{D}^m_k} \text{ } \& \text{ }
		{\mbLambda}^{emp}_k = \frac{1}{M} {\sum_{m=1}^M}   ( \mathcal{D}^m_k  -{{\mbmu_k^{emp}}}) ( \mathcal{D}^m_k  -{{\mbmu_k^{emp}}})^{\top}.  \label{emp}
	\end{align}
	Next, one can carry out Bayesian update that results in a-posteriori estimates of the mean and precision matrix
	for the solo Gaussian in \eqref{NIWp}, that is, 
	\begin{align} \label{Bayesian}
		\mbmu_k=\frac{{ k_0  \mbmu^0_k} + M \mbmu^{emp}_k  }{k_0+M}  ,\;
		{\mbLambda}_k =   \frac{{ (\mbLambda^0_k)^{-1} } + M\cdot \pmb{\Lambda}^{emp}_k + \kappa_k }{M+n_0} 
	\end{align}
	where $\kappa_k =   [k_0 M / (k_0 + M)] ({\mbmu^{emp}_k}-\mbmu^0_k) ({\mbmu^{emp}_k}-{\mbmu^0_k})^{\top}$.

	The Gaussian distribution  \eqref{NIWp} can be conditioned on states and action, i.e., $(\mathbf{x}_k,\mathbf{u}_k)$, using standard identities of multivariate Gaussians, which delivers
	the following parameters of \eqref{ltv_eq} i.e.,
	
	\begin{align*}
		\mathbf{A}^P_k = \begin{bmatrix}
			\mathbf{A}^d_k      & {\mathbf{B}^d_k} \\
			{\mathbf{A}^y_k}       & {\mathbf{B}^y_k} 
		\end{bmatrix} ,  \pmb{\Sigma}^P_k = {\begin{bmatrix}
				{\pmb{\Sigma}^d_k} & {{\pmb{\Sigma}^{yd}_k}} \\
				{\pmb{\Sigma}^{yd}_k}^\top      & { {\pmb{\Sigma}^y_k}}
		\end{bmatrix}  } .
	\end{align*}
	The dimensions of the matrices are  $\mathbf{A}^d_k \in \mathbb{R}^{n_s \times n_s}$, ${\mathbf{B}^d_k} \in \mathbb{R}^{n_s \times n_a}$, ${\pmb{\Sigma}^d_k} \in \mathbb{R}^{n_s \times n_s}$, ${{\mathbf{A}^y_k}} \in \mathbb{R}^{1 \times n_s}$, ${\mathbf{B}^y_k} \in \mathbb{R}^{1 \times n_a}$, ${ {\pmb{\Sigma}^y_k}} \in \mathbb{R}$, $ \mathbf{A}^P_k \in \mathbb{R}^{({n_s + 1}) \times (n_a +n_s) }$ and ${ {\pmb{\Sigma}^P_k} }  \in \mathbb{R}^{(n_s+1) \times (n_s+1)}$.
	{
		\begin{algorithm}[ht]
			\caption{DPSOC-EM algorithm}\label{alg:DPSOCEM}
			\begin{algorithmic}[1]
				{	
					
					\STATE  \textbf{Initialization:} Set $i=0$ and initialize $\hat{\phi}^0$ from one of the baselines (iLQG).
					
					\STATE   Run the real system under the controller \eqref{controlChap1}  with the controller parameter vector $\hat{\phi}^i$ to collect the data set  $\mathcal{ D}$; 
					\STATE Identify the dynamic model \eqref{ltv_eq} by fitting it to $\mathcal{ D}$ via VB inference.

					\STATE  Generate the cost observations ${\mathbb{Y}_{T}}$  using
					the dynamic model \eqref{ltv_eq}   
					and the controller  \eqref{controlChap1} with the controller parameter vector $\hat\phi^i$.

					\STATE  Perform the Kalman filter and R.T.S. smoother recursions to evaluate \eqref{eqn:KalmanValues}
					and hence  $ {\bar{\mathcal{L}} _k(\phi_k,\hat{\phi^i})}$, $k=1,\cdots, T$.
					
					\STATE Find $\hat{ \phi}_j^{i*} $, $j=T,\cdots 1$ according to \eqref{proposition:lastEq}.
					
					\STATE Update  $\hat{ \phi}^{i+1} = \hat{ \phi}^{i*} $. 
					
					\IF{$i+1 <$ number of recursions}
					\STATE   Let $i=i+1$; go to 2.
					\ELSE 
					\STATE Generate DDP solutions $\pi_{\mathcal{G}_1},\cdots, \pi_{\mathcal{G}_n}$
					\STATE Build initial sample set $ {S}$ from $\pi_{\mathcal{G}_1}, \cdots,\pi_{\mathcal{G}_n}, \pi_{\hat{\phi}^{i+1}}$  
					
					\FOR{number of iterations of iLQG $g \in \{0,\cdots,G\}$}   
					\STATE  Choose current sample set $S_g  \subset S$   
					\STATE Optimize $  \phi = \text{min}_\phi \mathbb{E}_{p_\phi }  \big[ Q(\mbs_k,\mba_k)  \big]  \text{ s.t. } D_{KL} (p_\phi   || p_{\hat{\phi}^{i+1}} ) \leq \nu_k  $ as per~\cite[Section~3.1]{levine2014learning} 
					\STATE Append samples from $\pi_{\phi}$ to $S_g$.
					\ENDFOR   
					\ENDIF
					\STATE Return the best policy $\pi_{\phi}$.
				}
			\end{algorithmic}
		\end{algorithm}
	}
	\section{DP Optimization of Mixture Likelihood} \label{sec:DPSOC}
	This section describes the details of the expression of the objective function which will be optimized, two theoretical results which are obtained by taking a similar approach to that of~\cite{mallick2022stochastic}.
	
	\subsection{Composite expression of mixture likelihood}

{This section introduces key concepts related to the widely recognized EM algorithm. The EM algorithm aims to estimate the parameter vector $\phi$ (which is customizable by the user) that maximizes the likelihood of an observed data set $\mathbb{Y}T$. Specifically, the likelihood of observing the data, denoted as $p_\phi(\mathbb{Y}_T)$, is iteratively maximized during the algorithm's execution as per,
\begin{equation} \label{theta_ml}
			\hat{\phi}_{EM} \in \{  \phi \in \Phi : p_{\phi} (\mathbb{Y}_{T}) \geq p_{\hat{\phi^i}} (\mathbb{Y}_{T}) \},
		\end{equation}
		where   
		$\hat{\phi^i}$ is a (known) considerably good parameter estimate with which the EM approach is initialized (at the
		iteration labeled $i$).

The algorithm ensures that the estimate $\hat{\phi}{EM}$ satisfies the condition expressed by Equation \ref{theta_ml}, where $\hat{\phi^i}$ represents a known and reasonably accurate initial parameter estimate at iteration $i$. This condition guarantees that the likelihood of the data $p{\phi} (\mathbb{Y}_{T})$ is not lower than the likelihood associated with $\hat{\phi^i}$.

The EM algorithm involves the computation of the \textit{observation log-likelihood} denoted as $L_\phi ({\mathbb{Y}_T})$, as defined in eq.~\ref{Lphi}. Additionally, it approximates the logarithm of the \textit{mixture likelihood} of latent variables ($\mathbb{S}_{T+1}$) and observations ($\mathbb{Y}_{T}$) using a surrogate function denoted as $\mathcal{L}(\phi,\hat{\phi}^{i})$, defined in Equation \ref{Ltheta_kheta}.

It is assumed that both $L_\phi ({\mathbb{Y}_T})$ and $\mathcal{L}(\phi,\hat{\phi}^{i})$ are differentiable with respect to $\phi\in\Phi$.} 
	{ 
		\begin{align}
			L_\phi ({\mathbb{Y}_T}) & {\triangleq \log p_\phi(  {\mathbb{Y}_{T}} )} \label{Lphi} 
		\end{align}
		\begin{align}
			\mathcal{L}(\phi,\hat{\phi}^{i})  &  \triangleq \mathbb{E}_{\hat{\phi}^i}   (  \log  p_\phi (\mathbb{S}_{T+1}, \mathbb{Y}_T) |\mathbb{Y}_T)  
			.  \label{Ltheta_kheta}
		\end{align}

{ 
Next, the primary objective is to maximize the mixture likelihood $\mathcal{L} (\phi,\hat{\phi}^i)$. The proposed optimization approach aims to find an improved policy parameter $\phi =\hat\phi^{i+1}$ for the next iteration that maximizes (or increases) $\mathcal{L} (\phi,\hat\phi^i)$ compared to $\phi =\hat\phi^{i}$.}
This will be done as per,
		\begin{align} \label{optiphi3}
			\hat{ \phi}^{i*}    = \arg\max_{\phi} \mathcal{L} (\phi, \hat\phi^i).
		\end{align}
		So, it is ideal to select $\hat\phi^{i+1} = \hat{ \phi}^{i*} $.
		The explicit expression of the mixture likelihood $\mathcal{L}(\phi, \hat{\phi^i})$  and a result presenting a connection between mixture likelihood and cumulative sum of cost-to-go metntioned in eq~\eqref{cumsum_obj1} is given in Lemma~\ref{lemma:LTV} and~\ref{thm:main_theorem} respectively. 	These results are adopted from~\cite{mallick2022stochastic} and are utilized to introduce an alternative approximate optimization framework that capitalizes on the underlying principles of the principle of optimality (refer to~\cite[Section 1.3]{bertsekas1995dynamic}). By employing this approach, optimal estimates of the control policy can be attained. A detailed description of the proposition, which relies on the principle of optimality, can be found in Proposition~\ref{propositionDPSOC}.
		\begin{lemma} \label{lemma:LTV}
The expression for the function $\mathcal{L}(\phi, \hat{\phi^i})$ considering the dynamic model \eqref{ltv_eq} and the controller \eqref{ltv_eq} is given by:   \begin{align} \label{calLbarcalL}
				\mathcal{L} (\phi,\hat{\phi^i})    =\log p(\mathbf{s}_1) +  \sum_{k=1}^T  \bar{\mathcal{L}}_k (\phi_k, \hat{\phi}^i),
			\end{align}
			for
			\begin{align} \label{ltv_surrogate}
				\log p(\mathbf{s}_1) =& -\frac{1}{2} \log |\mathbf{P}_1| + (\mathbf{s}_1 - \pmb{\mu}_1)^\top \mathbf{P}_1^{-1} (\mathbf{s}_1 - \pmb{\mu}_1) , \nonumber \\
				{\bar{\mathcal{L}} _k(\phi_k,\hat{\phi^i})} =&   -\frac{1}{2} \Tr \{ {{\pmb{\Sigma}}_k^o}^{-1}(  \mathbb{E}_{\hat\phi^i} ({\mbzeta_k \mbzeta_k ^\top | \mathbb{Y}_{T}) } -  \mathbb{E}_{\hat\phi^i} ({\mbzeta_k \mbz_k ^\top | \mathbb{Y}_{T}) } {\mathbf{A}^o_k}^{\top}  \nonumber\\
				   - {\mathbf{A}^o_k}  & \big[{ \mathbb{E}_{\hat\phi^i}  ({\mbzeta_k \mbz_k ^\top | \mathbb{Y}_{T}) }}\big]^{\top}  +  {\mathbf{A}^o_k}   \mathbb{E}_{\hat\phi^i} ({\mbz_k \mbz_k ^\top | \mathbb{Y}_{T}) }  {\mathbf{A}^o_k}^{\top}) \} -\frac{1}{2} \log  |{{{\pmb{\Sigma}}_k^o}}| ,
			\end{align}
			where
			{initial state $\mathbf{s}_1$ follows a Gaussian distribution with some known mean $\pmb{\mu}_1$ and covariance $\mathbf{P}_1$} and $\mbzeta_k = \col (\mbs_{k+1} ,\; y_k)$ and $\mbz_k = \col (\mathbf{s}_k, \mathbf{a}_k )$.
		\end{lemma}
		
		\begin{lemma}[Theorem IV.1 of~\cite{mallick2022stochastic}] \label{thm:main_theorem}
			Suppose the parameter $\hat{\phi}^{i+1}$ is produced such that
			\begin{align}  \label{calLi}
				\mathcal{L}( \hat{\phi}^{i+1}, \hat{\phi^i})  \geq    \mathcal{L} ( \hat{\phi^i} , \hat{\phi^i}).  \end{align}
			Then,  the cumulative sum of expected costs defined in \eqref{cumsum_obj} satisfies
			\begin{align} \label{EVi}
				\mathbb{E}_{ {\hat\phi^i}} ( V_{\hat\phi^{i+1}} ( {\mathbb{S}_{T+1}}) | \mathbb{Y}_T )  \leq 
				\mathbb{E}_{ {\hat\phi^i}} ( V_{\hat\phi^i} ( {\mathbb{S}_{T+1}}) | \mathbb{Y}_T ).
			\end{align}
		\end{lemma}

	}
	
}
\begin{proposition} \label{propositionDPSOC}
	For an initial state $s_1$, the optimal likelihood function $\mathcal{L}^*(s_1) \triangleq  \mathcal{L}( \hat{\phi}^{*}, \hat{\phi^i})$ of the above-mentioned problem (i.e., eq~\eqref{optiphi3}) is given by the last step of the following algorithm, which proceeds backwards in time from period $T$ to period 1. Therefore, the optimal parameters are obtained as a consequence of sequentially maximizing the likelihood backwards as presented in the following equations,
	
	\begin{align}  
		\hat{ \phi}_T^{i*}   
		& =  \arg\max_{\phi_T  }   \bar{\mathcal{L}}_T (\phi_T, \hat{\phi}^i) \label{eq1Proposition} \\
		\hat{ \phi}_{T-1}^{i*}   
		& =  \arg\max_{\phi_{T-1}  }    \bar{\mathcal{L}}_{T-1} (\phi_{T-1}, \hat{\phi}^i) +  \bar{\mathcal{L}}_{T} (\hat{\phi}_{T}^*, \hat{\phi}^i)  \label{eq2Proposition} \\
		\vdots \label{series} \\
		\hat{ \phi}_j^{i*}   
		& =  \arg\max_{\phi_{j}   } \bar{\mathcal{L}}_{j} (\phi_{j} , \hat{\phi}^i )   + \sum_{k=j+1 }^{T} \bar{\mathcal{L}}_k (\hat{\phi}^*_{k},  \hat{\phi}^i )  , \qquad j=T,\cdots, 2   \label{proposition:lastEq}
	\end{align}
	where 
	$  \bar{\mathcal{L}}_k (\phi_k, \hat{\phi}^i)$ is defined as in~\eqref{ltv_surrogate}   
	and the expectation inside each of the terms is taken with respect the to the measurement states which are a result of instantiations
	of the noise in the dynamical equation. In addition to that, if $a_k^* $ maximizes the term on the right in eq~\eqref{proposition:lastEq} for each $\mbs_k$ and $k$, the policy $\pi^*=\{\mba_1^*,\cdots,\mba^*_{T}\}$ is optimal.
\end{proposition}

\begin{remark}
	Precisely, there are two main reasons for proposing optimization according to eq~\eqref{proposition:lastEq},
	\begin{itemize}
		\item  Lemma~\ref{thm:main_theorem} establishes connections that likelihoods can be treated as costs with certain assumptions and approximations. Therefore, one can intuitively think of optimizing the approximate likelihood function from the last time step backwards in time in a very similar manner to that of the principle of optimality~\cite{bertsekas1995dynamic}.
		\item  The variance of the estimators of $\phi_j$ can be reduced by utilizing the fact that the calculation of future optimal parameters of actions mentioned in~\eqref{proposition:lastEq} does not take into account past likelihood functions unlike the objective function mentioned in~\cite[Section V-B]{mallick2022stochastic}. 
	\end{itemize}

\end{remark}

{ 
Attention is now directed towards the resolution of equation~\eqref{proposition:lastEq}. The proposed optimization paradigm aims to identify a superior policy parameter $\phi =\hat\phi^{i+1}$ for the next iteration compared to $\phi =\hat\phi^{i}$, with the objective of maximizing (or increasing) the right-hand side of equation~\eqref{proposition:lastEq}. Hence, it is desirable to choose $\hat\phi^{i+1} = \hat{ \phi}^{i*}$.

Nevertheless, computing the optimal value of $\hat{ \phi}^{i*}$ across the entire control action sequence ${\mathbf{a}_1, \mathbf{a}_2,.., \mathbf{a}_T}$ is typically challenging. The principle of expectation-maximization (EM) as an optimal control technique simplifies the maximization of ${\mathcal{L}} (\phi,\hat{\phi^i})$ over the entire control action sequence for each time step. In other words, one aims to maximize ${\mathcal{L}} (\phi,\hat{\phi^i})$ individually for each time step through a backward iterative process, as indicated by equations~\eqref{eq1Proposition}-\eqref{proposition:lastEq} for $j=T,\cdots,1$.}

Then, a better policy parameter for the next iteration is selected as 
$\hat{ \phi}^{i+1} = \col(\hat{ \phi}^{i*}_1, \cdots, \hat{ \phi}^{i*}_T) \approx \hat{ \phi}^{i*} $.
For brevity we  would refer the above optimization problem as \textbf{DPSOC-EM}.
After carrying out the maximization, we utilize the optimal parameters to collect samples of trajectories, that are utilized to carry out the dynamics fitting of piecewise linear models according to the procedure mentioned in~\cite{mallick2022reinforcement}. Then, with the help of those dynamics I carry out another degree of optimization as per:
\begin{align} \label{eqn:constrainequationEMandILQG}
	\phi = \text{min}_\phi \mathbb{E}_{p_\phi }  \big[ Q(\mbs_k,\mba_k)  \big]  \text{ s.t. } D_{KL} (p_\phi   || p_{\hat{\phi}^{i+1}} ) < \nu_k 
\end{align}
where $Q (s_k,a_k)$ is a second-order taylor series expansion employed around the nominal state-action pair ($\hat{\mbs}_k,\hat{\mba}_k$) as mentioned in~\cite[Chapter 3.4]{levine2014motor}. The term $p_{\hat{\phi}^{i+1}}$ in eq~\eqref{eqn:constrainequationEMandILQG} denotes the distribution of policy when excited with parameters $\hat{\phi}^{i+1}$. The term $\nu_k$ is the maximum \kl ~divergence between the new policy and the EM policy parameterized with $\hat{\phi}^{i+1}$. More details can be found in~\cite{levine2014learning,chebotar2017path}.

Having, described the new proposition, I present two results described in Theorem~\ref{cont3:uniqueness1} and~\ref{theorem:global_maximizerLastChap} that delivers some properties of the optimal estimates of~\textbf{DPSOC-EM}.

\begin{theorem} \label{cont3:uniqueness1} 
	The following equation,  
	\begin{align} \label{cont3:jacobian_set_zero}
		\nabla_{\phi_j} \Big\{  \bar{\mathcal{L}}_{j} (\phi_{j} , \hat{\phi}^i )   + \sum_{k=j }^{T} \bar{\mathcal{L}}_k (\hat{\phi}^*_{j+1},  \hat{\phi}^i )  \Big\}   = \mathbf{0}, \; j=T,\cdots, 1
	\end{align}
	has a unique solution for any given parameter  $\hat{\phi}^i$ for the function ${\bar{\mathcal{L}}_k (\phi_k,\hat{\phi^i})}$ defined in \eqref{ltv_surrogate}
\end{theorem} 
\begin{proof}
	The function $\bar{\mathcal{L}}_k (\phi_j, \hat{\phi}^i ) $ defined in \eqref{cont3:jacobian_set_zero} is expressed in terms of
	\begin{align}  \label{eqn:KalmanValues}
		\mathbb{E}_{{\hat{\phi}^i} } (y_k y_k^{\top}| {\mathbb{Y}_T}) ,\; \mathbb{E}_{{\hat{\phi}^i}} (\mathbf{a}_k \mathbf{a}_k^{\top}| {\mathbb{Y}_T}), \; \mathbb{E}_{{\hat{\phi}^i} } (y_k \mathbf{a}_k^{\top}| {\mathbb{Y}_T}),
	\end{align}  
	where $\mathbf{a}_k$ explicitly depends on ${\phi}_k$.
	Now, one can take a very similar approach as that of~\cite[Theorem IV.1]{mallick2022stochastic} to prove the following theorem, where components of expression of Jacobian can be described as follows,
	\begin{align}
		\nabla_{\phi_j} \Big\{  \bar{\mathcal{L}}_{j} (\phi_{j} , \hat{\phi}^i )   + \sum_{k=j }^{T} \bar{\mathcal{L}}_k (\phi^*_{j+1},  \hat{\phi}^i )  \Big\}   
		=  \nabla_{\phi_j } \Tr {\pmb{\lambda}_k}  + {\pmb{\mathcal{C}}}^\top    \label{equation_79Chap5}
	\end{align}
	where 
	\begin{align}
		{\pmb{\lambda}_k}  = {\pmb{\Sigma}^d_k}^{-1} \mathbf{B}^d_k \mathbb{E}_{\hat{\phi}^i} \big( \mathbf{a}_k \mathbf{a}_k^\top |\mathbb{Y}_T  \big) {\mathbf{B}^d_k}^\top     
	\end{align}
	and ${\pmb{\mathcal{C}}}$  represents some constant column vector,  independent of ${\phi}_j =\col (\mbf_j, \mbe_j, \mbsigma_j)$.
	In particular,  one has ${\pmb{\lambda}_k}= \pmb{\lambda}_{1,k}+\pmb{\lambda}_{2,k}+\pmb{\lambda}_{3,k} + \pmb{\lambda}_{4,k}+\pmb{\lambda}_{5,k}  $  with
	\begin{align*}
		\pmb{\lambda}_{1,k}&=(\pmb{\Sigma}^d_k)^{-1} \mathbf{B}^d_k \mathbf{F}_j \mathbf{G}_k \mathbf{F}_j^\top  {\mathbf{B}^d_k}^\top \\
		\pmb{\lambda}_{2,k}&=(\pmb{\Sigma}^d_k)^{-1} \mathbf{B}^d_k \mbe_j \mbe_j^\top {\mathbf{B}^d_k}^\top\\
		\pmb{\lambda}_{3,k}&=  (\pmb{\Sigma}^d_k)^{-1} \mathbf{B}^d_k \mbSigma_j {\mathbf{B}^d_k}^\top \\
		\pmb{\lambda}_{4,k}&= ({\pmb{\Sigma}^d_k})^{-1} \mathbf{B}^d_k \mathbf{F}_j \hat{\mathbf{s}}_{k|T} \mbe_j^\top {\mathbf{B}^d_k}^\top  \\
		\pmb{\lambda}_{5,k}&= ({\pmb{\Sigma}^d_k})^{-1} \mathbf{B}^d_k \mbe_j  \hat{\mathbf{s}}_{k|T}^\top \mathbf{F}_j^\top {\mathbf{B}^d_k}^\top .
	\end{align*}
	Also,  ${\pmb{\mathcal{C}}}$ is of the special structure
	\begin{align} \label{OstructureChap5}
		{\pmb{\mathcal{C}}} =   \begin{bmatrix}
			\sum_{k=j}^T \pmb{\mathcal{C}}_{1,k}\\
			\sum_{k=j}^T \pmb{\mathcal{C}}_{2,k}\\
			\mathbf{0}
		\end{bmatrix} , \end{align} 
	with the dimensions of $\pmb{\mathcal{C}}_{1,k}$, $ \pmb{\mathcal{C}}_{2,k}$
	and $\mathbf{0}$ corresponding to those of ${\mbf_{j}}$, ${\mbe_{j}}$ and  
	$\mbsigma_{j}$, respectively. The explicit expression of  
	$ \pmb{\mathcal{C}}_{1,k}$ and $ \pmb{\mathcal{C}}_{2,k}$ can be obtained from the equations in~\cite[Appendix A]{mallick2022stochastic}.
	
	Below, we calculate the derivative of the terms in \eqref{equation_79Chap5} with respect to 
	$\mbf_j$, $\mbe_j$, and $\mbsigma_j$, respectively. 
	
	Firstly, with respect to  $\mbf_j$, one has 
	\begin{align*}
		\nabla_{\mbf_j}   \Tr  \pmb{\lambda}_{1,k} =\mbf_j^\top  \pmb{\mathcal{M}} ^{1}_k ,\;
		\nabla_{\mbf_j}  \Tr  \pmb{\lambda}_{2,k} =  0,\;
		\nabla_{\mbf_j}    \Tr  \pmb{\lambda}_{3,k} = 0
	\end{align*}
	and hence 
	\begin{align*}
		& \nabla_{\mbf_j}     \Tr  \pmb{\lambda}_{4,k} \\
		&=\nabla_{\mbf_j}    \Tr  \Big\{ ({\pmb{\Sigma}^d_k})^{-1} \mathbf{B}^d_k \mathbf{F}_j \hat{\mathbf{s}}_{k|T} \mbe_j^\top {\mathbf{B}^d_k}^\top \Big\} \\ 
		&=\nabla_{\mbf_j}     {\Tr  \Big\{ (\mbe_j^\top {\mathbf{B}^d_k}^\top) ({\pmb{\Sigma}^d_k}^{-1} \mathbf{B}^d_k \mathbf{F}_j \hat{\mathbf{s}}_{k|T}  ) \Big\} } \\
		&=  (\mathbf{B}^d_k \mbe_j)^\top \nabla_{\mbf_j}  (\hat{\mathbf{s}}_{k|T}^\top \otimes { \pmb{\Sigma}_k^d}^{-1} \mathbf{B}^d_k)
		\mbf_j \\
		&  =   [(\hat{\mathbf{s}}_{k|T} \otimes {  {\mathbf{B}^d_k}^\top {\pmb{\Sigma}_k^d}^{-1}}^\top )  (\mathbf{I} \otimes {\mathbf{B}^d_k} )
		\mbe_j ]^\top \nonumber 
	\end{align*}
	and, similarly, 
	\begin{align*}
		\nabla_{\mbf_j}    \Tr  \pmb{\lambda}_{5,k}  = [ (\hat{\mathbf{s}}_{k|T} \otimes {\mathbf{B}^d_k}^\top )  (\mathbf{I} \otimes  {  {\pmb{\Sigma}_k^d}^{-1}}  {{\mathbf{B}^d_k}}  )  \mbe_j]^\top.
	\end{align*}
	Here,   $\pmb{\mathcal{M}} ^1_k = \pmb{\mathcal{M}}_k^{1,0} +  \pmb{\mathcal{M}} ^{1,1}_k$ with
	\begin{align*}
		\pmb{\mathcal{M}} _k^{1,0}   &= 2  \hat{\mathbf{s}}_{k|T} \hat{\mathbf{s}}_{k|T}^\top  \otimes  {\mathbf{B}^d_k}^\top {\pmb{\Sigma}^d_k}^{-1} \mathbf{B}^d_k   \\
		\pmb{\mathcal{M}} _k^{1,1} &= 2  \hat{\mathbf{P}}_{k|T}   \otimes  {\mathbf{B}^d_k}^\top {\pmb{\Sigma}^d_k}^{-1} \mathbf{B}^d_k  .
	\end{align*}
	
	Secondly, with respect to  $\mbe_j$, one has 
	\begin{align*}
		\nabla_{\mbe_j}   \Tr  \pmb{\lambda}_{1,k} &= 0\\
		\nabla_{\mbe_j}   \Tr  \pmb{\lambda}_{2,k} &= \mbe_j^\top \pmb{\mathcal{M}} ^{2}_k \\
		\nabla_{\mbe_j}   \Tr  \pmb{\lambda}_{3,k} &=   0 \\
		\nabla_{\mbe_j}  \Tr  \pmb{\lambda}_{4,k} &= 
		\mbf_j^\top (\mathbf{I} \otimes {\mathbf{B}^d_k}^\top) (\hat{\mathbf{s}}_{k|T} \otimes { \pmb{\Sigma}_k^d}^{-1} \mathbf{B}^d_k ) \\
		\nabla_{\mbe_j}   \Tr  \pmb{\lambda}_{5,k}  &= \mbf_j^\top (\mathbf{I} \otimes {\pmb{\Sigma}_k^d}^{-1} {\mathbf{B}^d_k})^\top (\hat{\mathbf{s}}_{k|T} \otimes \mathbf{B}^d_k ) 
	\end{align*}
	with
	\begin{align*}
		\pmb{\mathcal{M}} ^2_k  &=2    \mathbf{I} \otimes  {\mathbf{B}^d_k}^\top {\pmb{\Sigma}^d_k}^{-1} \mathbf{B}^d_k   . 
	\end{align*}

	From above, the equation \eqref{cont3:jacobian_set_zero} is equivalent to    
	\begin{align} \label{optimal_estimatesChap3}
		\nabla_{\phi_j} \Big\{  \bar{\mathcal{L}}_{j} (\phi_{j} , \hat{\phi}^i )   + \sum_{k=j+1}^{T} \bar{\mathcal{L}}_k (\phi^*_{j+1},  \hat{\phi}^i )  \Big\}    =
		\phi_j^\top   \begin{bmatrix}
			\pmb{\mathcal{M}}  & \mathbf{0}\\
			\mathbf{0} & \pmb{\mathcal{M}}^2 
		\end{bmatrix}        +   \pmb{\mathcal{C}}^\top  = \mathbf{0}
	\end{align}
	where   $\pmb{\mathcal{M}}^2 = \sum_{k=1}^T \pmb{\mathcal{M}}^2_k$, $\pmb{\mathcal{M}}  = \sum_{k=1}^T \pmb{\mathcal{M}}_k$, 
	$\pmb{\mathcal{M}}_k=  \begin{bmatrix}
		\pmb{\mathcal{M}}_k^1   &  \pmb{\mathcal{M}}_k^3  \\
		{\pmb{\mathcal{M}}_k^3}^\top &    {\pmb{\mathcal{M}}_k^2 }
	\end{bmatrix}$ with $
	\pmb{\mathcal{M}} ^3_k  = 2 (  \mathbf{I}_{1} \otimes {\mathbf{B}^d_k}^\top) (\hat{\mathbf{s}}_{k|T}^\top \otimes {\pmb{\Sigma}^d_k}^{-1} {\mathbf{B}^d_k}).
	$

	What is left is to prove the existence of a unique solution $\phi_j$ to the equation
	\eqref{optimal_estimatesChap3}. It suffices to show that  $\pmb{\mathcal{M}}  >0$
	and  $ \pmb{\mathcal{M}}^2>0$, or $\pmb{\mathcal{M}}_k >0$
	and  $ \pmb{\mathcal{M}}^2_k>0$.

	Since the matrix $\mathbf{B}^d_k$ has a  full column rank, $\pmb{\Sigma}^d_k >0$ and $\hat{\mathbf{P}}_{k|T} > 0$, one has
	$\pmb{\mathcal{M}} ^2_k  >0$ and  $\pmb{\mathcal{M}} _k^{1,1} >0$. 
	Next,   the decomposition of the matrix $\pmb{\mathcal{M}}_k$ gives
	\begin{align} \label{decomposition3}
		\pmb{\mathcal{M}}_k  = & \begin{bmatrix}
			\mathbf{I} &  \pmb{\mathcal{M}} ^3_k  {\pmb{\mathcal{M}} ^2_k}^{-1} \\
			\mathbf{0} & \mathbf{I}
		\end{bmatrix} 
		\begin{bmatrix}
			\pmb{\mathcal{M}} _k^{1}    -  \pmb{\mathcal{M}} _k^3 {\pmb{\mathcal{M}} _k^2}^{-1}  {\pmb{\mathcal{M}} _k^3}^\top   &  \mathbf{0} \\
			\mathbf{0} &  \pmb{\mathcal{M}} _k^2
		\end{bmatrix} &  \\
		& \qquad \qquad \qquad  \qquad \qquad \qquad \qquad \begin{bmatrix}
			\mathbf{I} & \mathbf{0}  \\
			{\pmb{\mathcal{M}} ^2_k}^{-1} {\pmb{\mathcal{M}} ^3_k}^\top   & \mathbf{I}
		\end{bmatrix}.
	\end{align}
	It is noted that $\pmb{\mathcal{M}}_k ^{1,0} = \pmb{\mathcal{M}}_k^3 {\pmb{\mathcal{M}}_k^2}^{-1}  {\pmb{\mathcal{M}}^3_k}^\top $, 
	which implies
	\begin{align*}
		\pmb{\mathcal{M}} _k^{1}    -  \pmb{\mathcal{M}} _k^3 {\pmb{\mathcal{M}} _k^2}^{-1}  {\pmb{\mathcal{M}} _k^3}^\top
		=\pmb{\mathcal{M}}_k ^{1,1} >0
	\end{align*}
	and hence $\pmb{\mathcal{M}}_k >0$. The proof is thus completed.
\end{proof}
\begin{figure}[H]
	\centering
	{\includegraphics[scale=0.248]{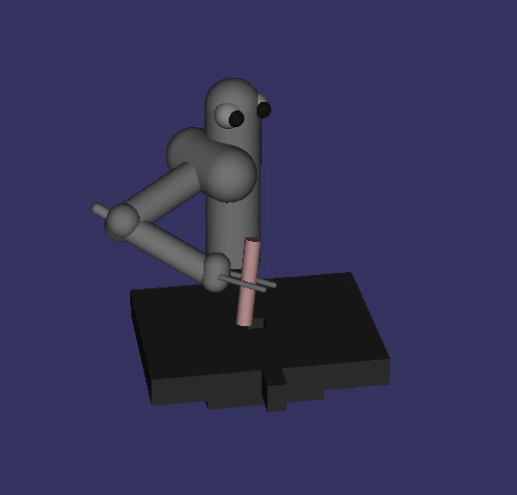}}
	\hfill
	{\includegraphics[scale=0.134]{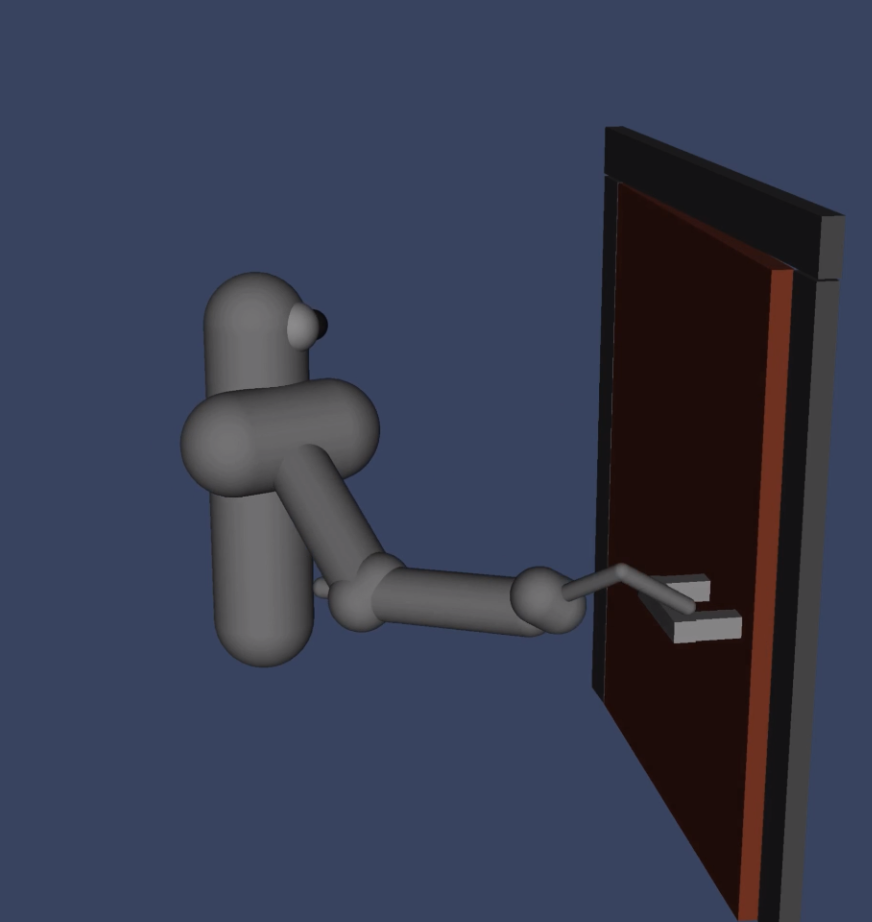}}
	\caption{\texttt{MUJOCO}-based simulation environment on which the DP-based stochastic optimal control algorithm is tested; \textbf{left:} A \texttt{PR2} robot has the task of inserting the peg in the hole; \textbf{right: }  A \texttt{PR2} robot has the task of opening the door to a certain minimum angle.}
	\label{fig:contribution3MujocoSnaps}
\end{figure}
\begin{figure}[H]
	\centering
	\includegraphics[width=\columnwidth]{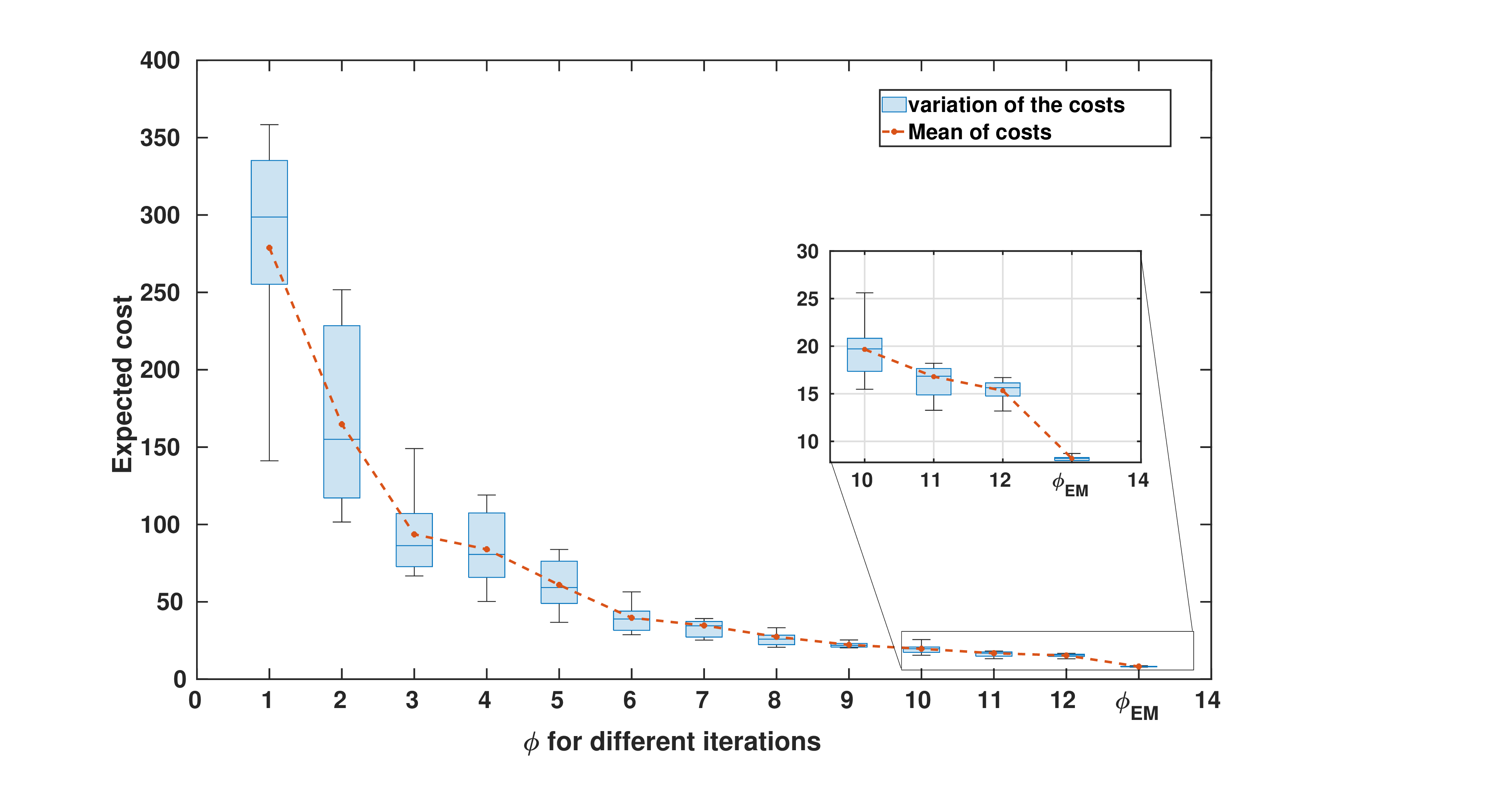}%
	\caption{Expected costs for peg in the hole task; the zoomed-in plot shows the decline in the cost when $\phi_{12}$ is utilized to initiate the EM strategy, i.e., we used iLQG to initiate EM when I do not see much change in the previous costs.}
	\label{fig:contributionStatesPeg}
\end{figure}
\begin{theorem} \label{theorem:global_maximizerLastChap}
	For the function ${\bar{\mathcal{L}}_k (\phi_k,\hat{\phi}^i)}$ defined in \eqref{ltv_surrogate},  the following inequality  
	\begin{align}
		& - \nabla^2_{\phi_j} \big[ \bar{\mathcal{L}}_{j } (\phi_{j}, \hat{\phi}^i )   + \sum_{k=j }^{T} \bar{\mathcal{L}}_k (\hat{\phi}^*_{j+1}, \hat{\phi}^i ) \big] > 0, \; j=T,\cdots, 1  \label{total_matris_kraceLastChap}
	\end{align}
	always holds for any given parameter $\hat{\phi^i}$; $\hat{\phi}^*_{j+1}$ represents the optimal parameters obtained from the backward steps of optimization as per the series of equations~\ref{series}.
\end{theorem} 
\begin{figure*}[ht]
	\centering
	{\includegraphics[scale=0.35]{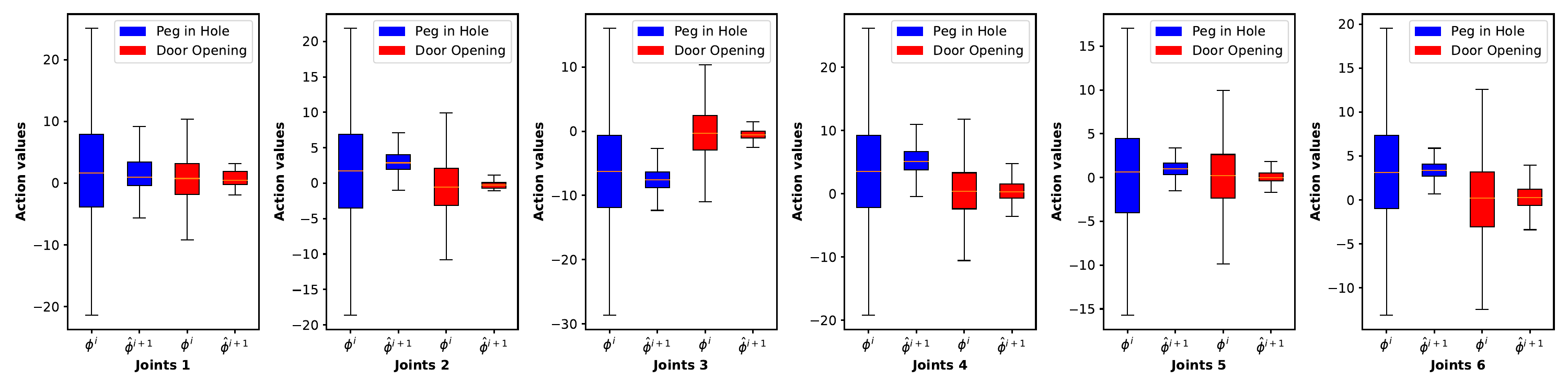}}
	\caption{The plots represents the over all variation of control actions for $40$ sample paths and each sample path contains samples of horizon length T for each of the $6$ DOF of \texttt{PR2} robotic arm joints evaluated on two tasks corresponding to the labels of colors red and blue. The term $\phi^i$ and $\hat{\phi}^{i+1}$ corresponds to baseline (iLQG) and improved parameters (EM) respectively.}
	\label{fig:contribution3ControlNoise}
\end{figure*}
\begin{proof}
By combining the insights from Theorem~\ref{cont3:uniqueness1} with the proof presented in~\cite[Theorem 3]{mallick2022stochastic}, one can directly demonstrate the theorem's validity.
\end{proof}
\begin{remark}
	The two differences between the objective function in eq~\eqref{proposition:lastEq} and the objective function mentioned in~\cite[Section 4B]{mallick2022stochastic} are:
	\begin{enumerate}
		\item   For each time-varying optimization, there is the total summation across $j^{th}$ to $T^{th}$ time instants rather than all time instants. 
		\item A backward substitution of optimal parameters is utilized as per eq~\eqref{eq2Proposition} that leverages optimal parameters from eq~\eqref{eq1Proposition}.
	\end{enumerate} 
	Nevertheless, the functional form which governs the optimization is essentially the same. As a result one can take very similar approach to carry out theoretical developments as per~\cite{mallick2022stochastic,mallick2022reinforcement}.
\end{remark}
\begin{figure*}[ht]
	\begin{subfigure}{\figsmallwidth}
		\includegraphics[width=\columnwidth]{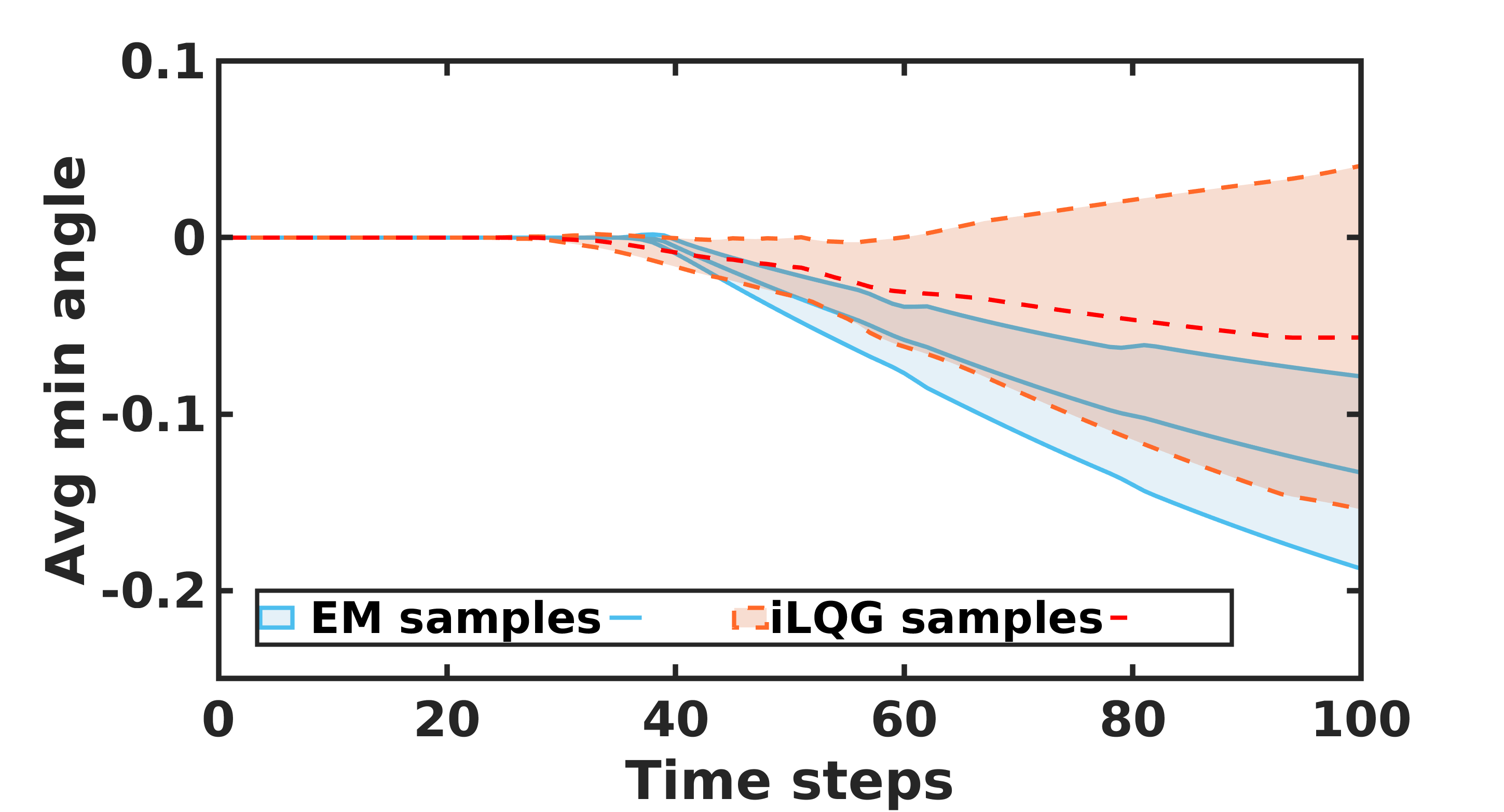}%
	\end{subfigure} 
	\begin{subfigure}{\figsmallwidth}
		\includegraphics[width=\columnwidth]{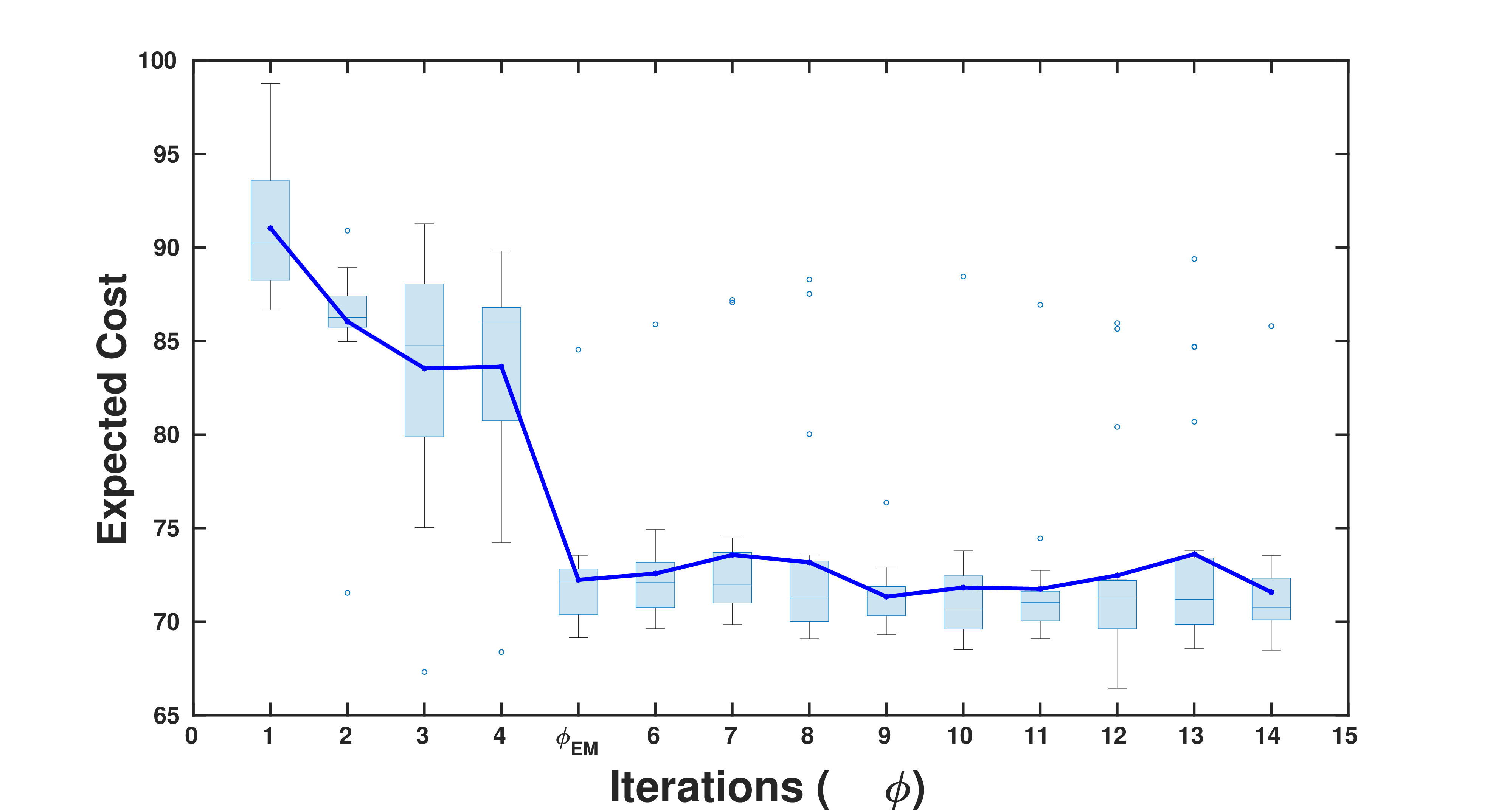}%
	\end{subfigure} 
	\begin{subfigure}{\figsmallwidth}
		\includegraphics[width=\columnwidth]{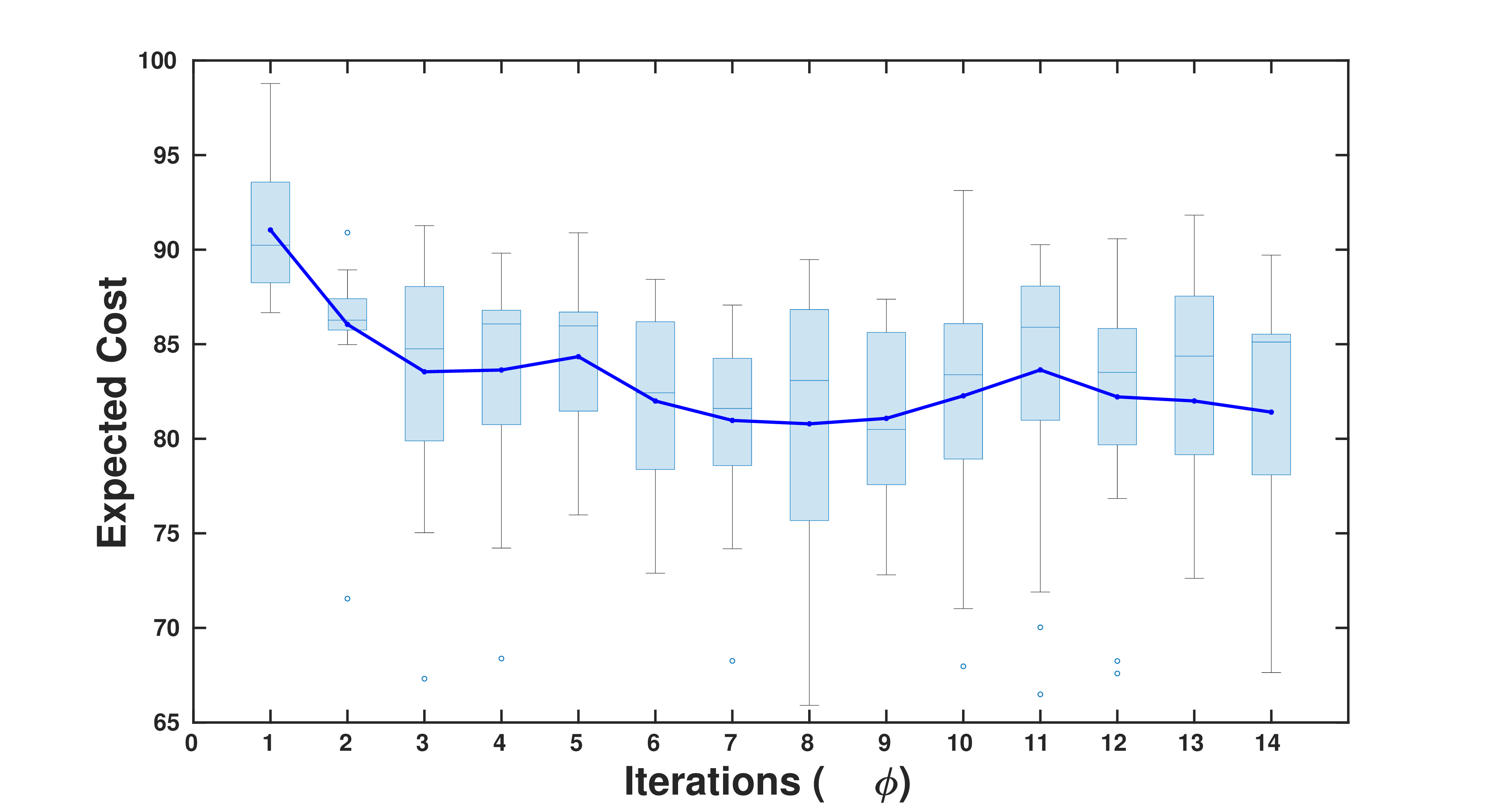}%
	\end{subfigure}
	\caption{ \textbf{Top: } Average minimum angle in radians of the door hinge opened by the hook attached to the end effector of the \texttt{PR2} robot (lower angle is better); \textbf{Bottom left: } Expected costs for door opening task based on Algorithm~\ref{alg:DPSOCEM}; \textbf{Bottom right: } Expected costs for door opening task for 14 iterations based on parameter obtained from iLQG.}
	\label{fig:contribution3CostILQGEMILQG}
\end{figure*}
 \section{New Interpretations/Connections of EM-based SOC with Legendre Transform} \label{sec:connection}
	{

		This section aims at providing an information-theoretic perspective on the EM algorithm by highlighting its relationship with optimal control and its connections to earlier findings in the literature of stochastic control theory and dynamic games~\cite{dai1996connections}.
		
		An investigation conducted by~\cite{theodorouinformation} has yielded a comprehensive theoretical analysis, wherein they have associated the concept of \enquote{free energy} with the formulation of optimal control. Employing the definition of free energy as a foundational premise, the authors exploit the Legendre transformation mentioned in Lemma~\ref{lemma:radon}, thereby delivering a pivotal theorem that is presented in Theorem~\ref{theorem:information_theory}. The inclusion of the following lemma serves to ensure self-containedness, enhance the overall coherence of this section and later this is used to deliver a theoretical result.
		
  \begin{lemma} \label{lemma:radon}
			Consider $(\Omega, \mathcal{F})$ to be a measurable space, where $\Omega$ denotes the sample space and $\mathcal{F}$ denotes a $\sigma-$algebra, and let $\mathcal{P}(\Omega)$ define a probability measure on the $\sigma-$algebra $\mathcal{F}$. Consider $\mathbb{P},\mathbb{Q} \in \mathbf{\pmb{P}} (\Omega)$ and the definitions of free energy and relative entropy from Definitions~\ref{def:FreeEnergy} and~\ref{def:Entropy} and under assumption $\mathbb{Q}<<\mathbb{P}$, the following equality holds:
			\begin{align} \label{inf_sup3}
				& \text{   } [\frac{1}{\rho} \log \int e^{\rho\mathcal{J}(\mathbb{S}_T,\mathbb{A}_T)} d\mathbb{P}  ]    =  \  \text{   }\sup_{\mathbb{Q}} [ \mathbb{E}_{\mathbb{Q}} \mathcal{J}(\mathbb{S}_T,\mathbb{A}_T) - \frac{1}{\rho} \text{D}_{\text{KL}} (\mathbb{Q}||\mathbb{P})) ], 
			\end{align}
			where $\mathcal{J}$ is a measurable function, i.e., $\mathcal{J}(\cdot):\Omega \rightarrow \mathbb{R}$ and can be thought of as a cost function.\footnote{This cost function may be obtained as a result of controlled or uncontrolled dynamics as per~\cite{theodorouinformation}.}
		\end{lemma}
		\begin{proof}
			One can denote the expression of expectation $\mathbb{E}_{\mbbP}$ as a function of the expectation $\mathbb{E}_{\mbbQ}$ as per Definition~\ref{def:Entropy}. Precisely, 
			\begin{align} 
				\mathbb{E}_{\mbbP}\Big[ e^{\rho \mathcal{J}} \Big] &= \int e^{\rho \mathcal{J} } {d\mbbP} = \int e^{\rho \mathcal{J} } \frac{d\mbbP}{d\mbbQ}  {d\mbbQ} \\
				\implies \log \mathbb{E}_{\mbbP}\Big[ e^{\rho \mathcal{J}} \Big] &= \log \Big[ \int e^{\rho \mathcal{J}}  \frac{d\mbbP}{d\mbbQ}  {d\mbbQ} \Big] \geq \int \log \Big(  e^{\rho \mathcal{J} }  \frac{d\mbbP}{d\mbbQ}   \Big)  {d\mbbQ} \label{eq:jensen} \\
				\implies \log \mathbb{E}_{\mbbP}\Big[ e^{\rho \mathcal{J} (\mathbb{S},\mathbb{A}) } \Big] & \geq \int \Big( {\rho \mathcal{J}} + \log \frac{d\mbbP}{d\mbbQ} \Big) {d\mbbQ} = \int      {\rho \mathcal{J} }     {d\mbbQ}  - D_{KL} (\mbbQ|| \mbbP)  \label{finaleqLegendreTransform}
			\end{align}
		Eq.~\eqref{eq:jensen} is Jensen's inequality applied to terms of logarithm of integrals. Then, multiplying eq~\eqref{finaleqLegendreTransform} with $\frac{1}{\rho}$, and applying $\mathbb{E}_{\mbbQ}(\mathcal{J}) = \int \mathcal{J} d \mbbQ $, and changing signs, one can prove eq~\eqref{inf_sup3}.
		\end{proof}
  \begin{remark}
      Above lemma comes from~\cite[Theorem 1]{theodorou2015nonlinear} and the value of $\rho<0$ describes risk sensitive attribute of the optimization. For our analysis, we will associate $\rho<0$ with a parameter of distribution that governs the costs/rewards observations.
  \end{remark}
		The next theorem exploits the relationship between EM-based optimal control according to eq~\eqref{optiphi3} and results of Lemma~\ref{lemma:radon} to provide a novel interpretation.
		\begin{theorem} \label{theorem:information_theory}
			Assuming $\rho=-(\lambda-1)$, where $\lambda$ is a parameter as per Assumption~\ref{assumption:pdf_Yk}, the EM-based optimal control approach is a special case of infimum operation on the free energy (as per Definition~\ref{def:FreeEnergy})) i.e., 
			\begin{align}
			 \inf_{\mathbb{A}_T} \Big[ \log \int e^{( \rho \sum_{k=1}^T Y_k(\mathbf{s}_k,\mba_k) )} \Big] =  \sup_\phi \mathcal{L}(\phi,\hat{\phi}^i)
			\end{align}
		\end{theorem}
		\begin{proof}
			Exploiting Lemma~\ref{lemma:radon} and substituting, $ \mathbb{P} = p_{\hat{\phi}^i}(\mathbb{S}_{T+1}|\mathbb{Y}_T)$ and $\mathbb{Q}= \widetilde{p}(\mathbb{S}_{K})$ and utilizing~Definition~\ref{def:FreeEnergy}, i.e., $ \mathcal{E} \triangleq \mathbb{E}( \mathcal{J}(\mathbb{S}_T,\mathbb{A}_T))) = \log \int e^{( \rho \mathcal{J}(\mathbb{S}_T,\mathbb{A}_T) )} $, where $\mathcal{J}(\mathbb{S}_T,\mathbb{A}_T) \triangleq \sum_{k=1}^T Y_k(\mbs_k,\mba_k) $, thus upon expansion of the free energy term and exploiting the Jensens's inequality,
			\begin{align}
				& \log \int e^{( \rho \sum_{k=1}^T Y_k(\mathbf{s}_k,\mba_k) )}\\&  \geq \int \log \Big(e^{\rho  \mathcal{J}(\cdot)} \frac{dp_{\hat{\phi}^i}(\mathbb{S}_{T+1}|\mathbb{Y}_T)}{d\widetilde{p}(\mathbb{S}_T)}  \Big) {d\widetilde{p}(\mathbb{S}_T)} \\
				\implies & \log \int e^{( \rho \sum_{k=1}^T Y_k(\mathbf{s}_k,\mba_k) )} \nonumber \\
				&    \geq  \mathbb{E}_{\widetilde{p}(s_k)}  \Big[ \log \frac{dp_{\hat{\phi}^i}(\mathbb{S}_{T+1}|\mathbb{Y}_T)}{d\widetilde{p}(\mathbb{S}_T)}  + \rho \mathcal{J} ((\mathbb{S}_T,\mathbb{A}_T)) \Big] \\
				\implies &  \log \int e^{( \rho \sum_{k=1}^T Y_k(\mathbf{s}_k,\mba_k) )} \nonumber\\
				&   \geq -\text{D}_{\text{KL}} (\widetilde{p}(\mathbb{S}_k) || p_{\hat{\phi}^i} ((\mathbb{S}_{T+1}|\mathbb{Y}_T))  ) + \mathbb{E}_{\widetilde{p}(s_k)} [\rho \mathcal{J}((\mathbb{S}_T,\mathbb{A}_T))]  \\
				 \implies & \mathcal{E}(\cdot)  = \sup_{d\mathbb{Q}} \Big[- \text{D}_{\text{KL}}  (\widetilde{p}(s_k) || p_{\hat{\phi}^i} (\cdot)  ) + \mathbb{E}_{\widetilde{p}(s_k)} \rho \mathcal{J}(\cdot)   \Big] ,
			\end{align}
			Essentially, one can define the minimization problem as:
			\begin{align}
				& \log \int e^{( \rho \sum_{k=1}^T Y_k(s_k,a_k) )}  \\ &= \sup_{d\mathbb{Q}} \Big[ - \text{D}_{\text{KL}}  ( \mathbb{Q} || p_{\hat{\phi}^i} (\cdot)  ) + \mathbb{E}_{\widetilde{p}(s_k)} \rho \mathcal{J}(\cdot)   \Big] \\
				&   \implies \mathcal{E}(\cdot) =  \sup_{d\mathbb{Q}} \Big[ - \text{D}_{\text{KL}}  ( \mathbb{Q} || p_{\hat{\phi}^i} (\cdot)  ) + \mathbb{E}_{\widetilde{p}(s_k)} \rho \mathcal{J}(\cdot)   \Big] , \label{e_step_free_energy}
			\end{align}
			where the $(\cdot)$ is used for shortness of expression. The eq.~\eqref{e_step_free_energy} is of a very similar structure as that of eq.~\eqref{Lemma_1_property}, which shows a special case where the free energy definition can be reduced to E-step for optimal policy. 
		Further, assuming $\rho=-(\lambda-1)$, one can see,
			\begin{align}
			\mathcal{E} (\cdot)	= &  \inf_{d\mathbb{Q}}  \Big[ \text{D}_{\text{KL}}  (\widetilde{p}(s_k) || p_{\hat{\phi}^i} (\mathbb{S}_{T+1}|\mathbb{Y}_T)  )   + \mathbb{E}_{\widetilde{p}(s_k)} (\lambda-1) \mathcal{J} (\cdot) \Big].
			\end{align}
	If  $\widetilde{p} (\mathbf{s}_k) =p_{\hat{\phi}^i} (\mathbb{S}_{T+1}|\mathbb{Y}_T)  )  $, then, one can minimize the free energy $\mathcal{E}(\cdot)$ on set of control action $\{\mba_1, \cdots, \mba_T\}$. Without loss of generality, the set of control actions $\{\mba_1, \cdots, \mba_T\}$ can be represented in terms of parameters $\phi$, therefore,
			\begin{align}
		\inf_{\mba_1,\cdots,\mba_T} \mathcal{E}(\cdot)	& =	\inf_{\mba_1,\cdots,\mba_T}  \Big[  \inf_{d\mathbb{Q}} \big[ \mathbb{E}_{\mathbb{Q}} (\lambda-1) \mathcal{J} (\cdot) \big] \Big]  \\
    & =\underbrace{\sup_{\phi} \underbrace{\inf_{d\mathbb{Q}}  \Big[ l(\phi , \mathbb{Q} ) \Big]}_{\text{E-step of EM-SOC}}}_{\text{M-step of EM-SOC}}  \\
    & =  \sup_{\phi} \mathcal{L}(\phi, \hat{\phi}^i). \label{lasteq_connection}
			\end{align} 
   The eq.~\eqref{lasteq_connection} is a direct result of infimum operation over a set of $\phi$ on eq.~\eqref{eq:28_new}.
		\end{proof}
	}

	\begin{remark}
	The parameter $\rho$ plays the similar role as that of $-(\lambda-1)$ and this  theorem establishes that EM-based optimal control according to eq.~\eqref{optiphi3} is a a special case of minimization of free energy-relative entropy. Also, it can be noted that the expectation operator is taken with respect to some p.d.f. of the states which evolve according to the controlled dynamics (as per the linear Gaussian model \eqref{ltv_eq}), whereas \cite[Section 3]{theodorouinformation} deals both with controlled and uncontrolled dynamics. Nevertheless, this specialised reduction of optimization operations under some assumptions can be reduced to be equivalent to optimization . This also provides novel interpretation of the duality between the already existing information-theoretic-based min-max SOC approach (from Eq. \eqref{inf_sup3}) and that of the proposed objective function (according to Eq. \eqref{optiphi3})
	\end{remark}
}
\begin{figure*}[t]
		\begin{subfigure}{\figsmallwidth}
			\includegraphics[width=\columnwidth]{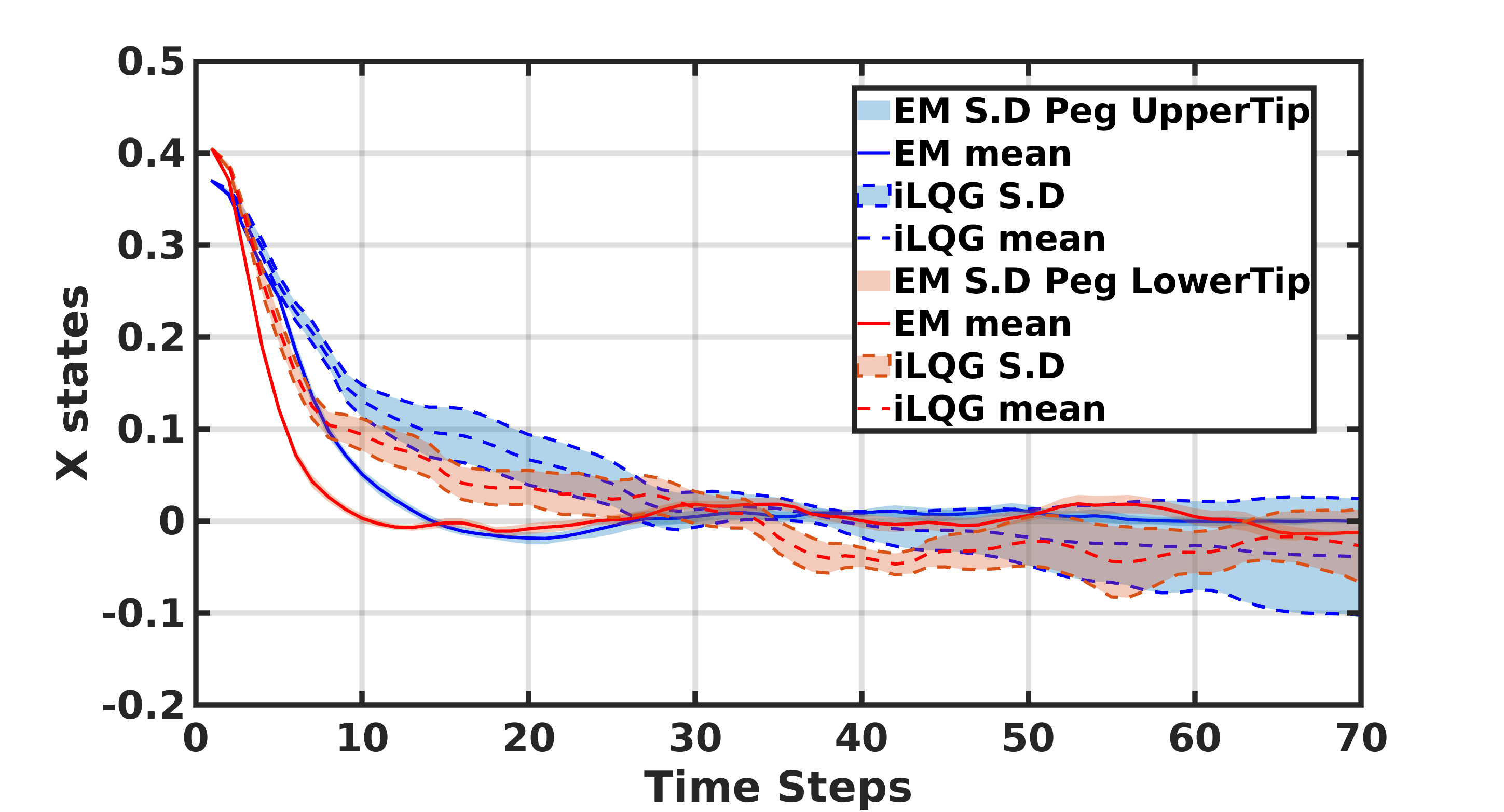}%
		\end{subfigure}
		\begin{subfigure}{\figsmallwidth}
			\includegraphics[width=\columnwidth]{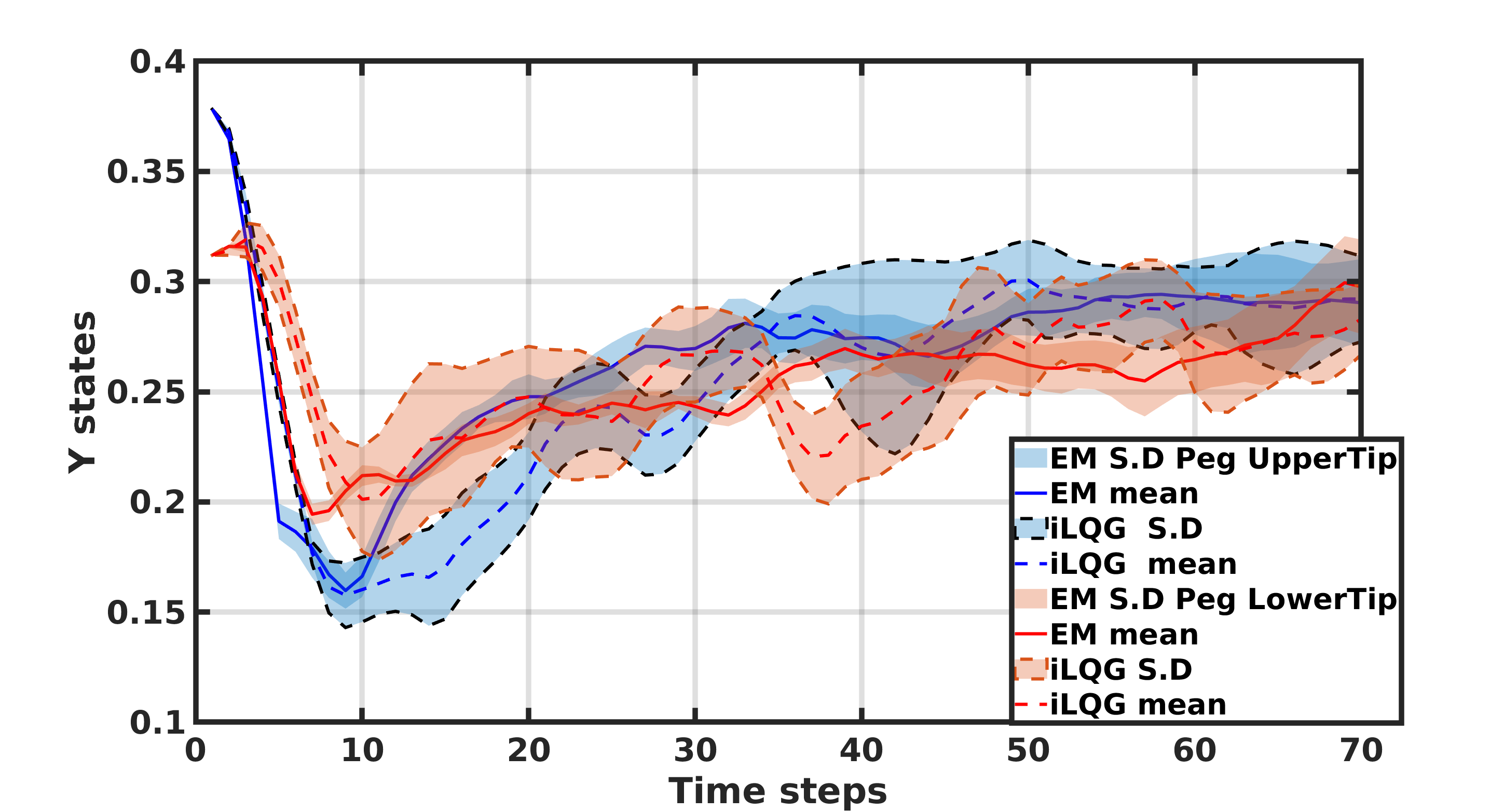}%
		\end{subfigure} 
		\begin{subfigure}{\figsmallwidth}
			\includegraphics[width=\columnwidth]{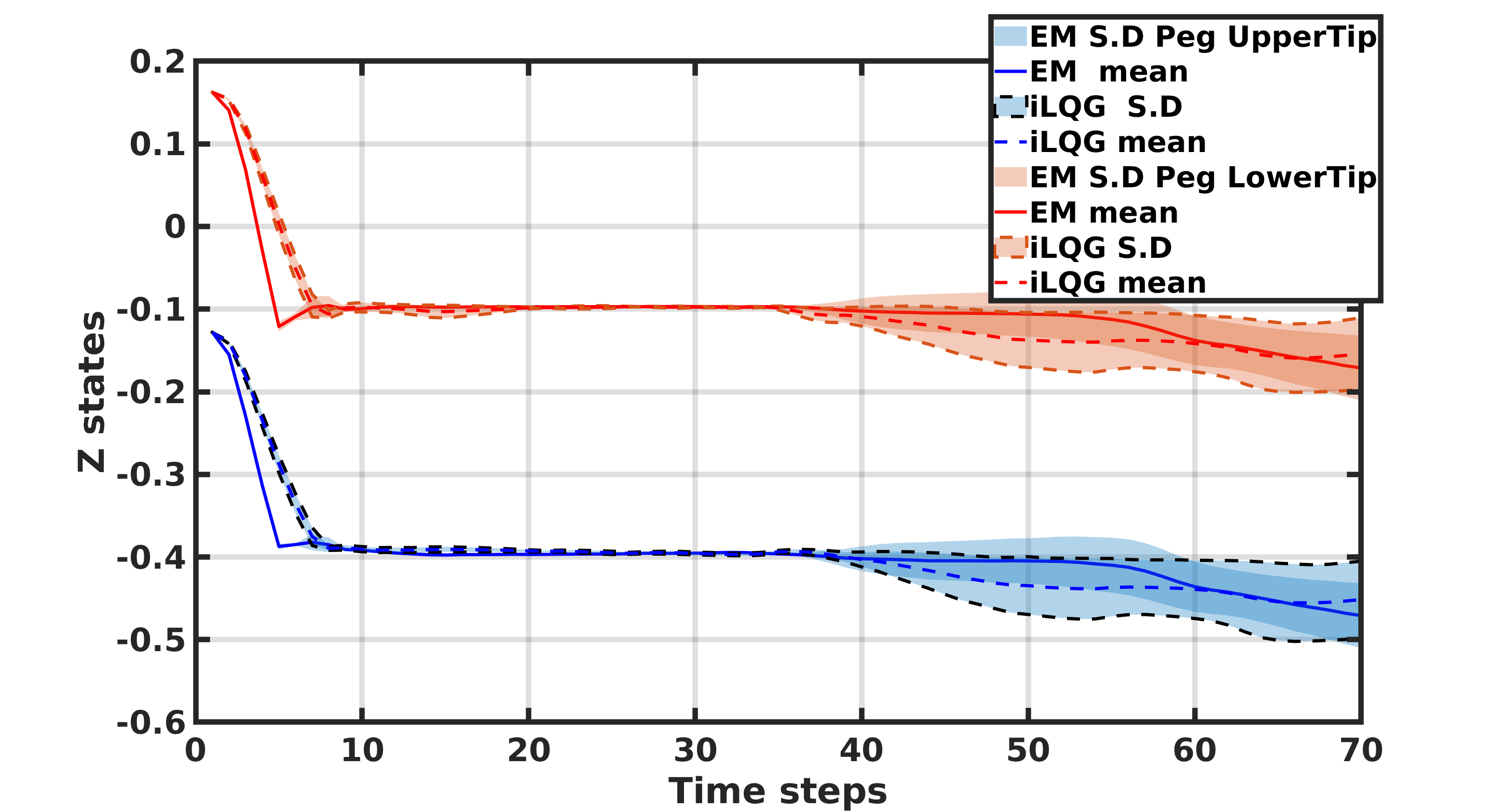}%
		\end{subfigure}
		\caption{ Three states i.e., $[x,y,z]$ positions of the peg which is connected to the end-effector of $\texttt{PR2}$ robot. Each plot of this figure is divided into two coloured plots that shows the mean and variance of the state trajectories. The dashed lines represents iLQG trajectories and solid lines represents that of parameters obtained from eq~\eqref{proposition:lastEq}.}
		\label{fig:contribution3StateNoise}
	\end{figure*}
 	\begin{figure*}[t]
		\begin{subfigure}{\figsmallwidth}
			\includegraphics[width=\columnwidth]{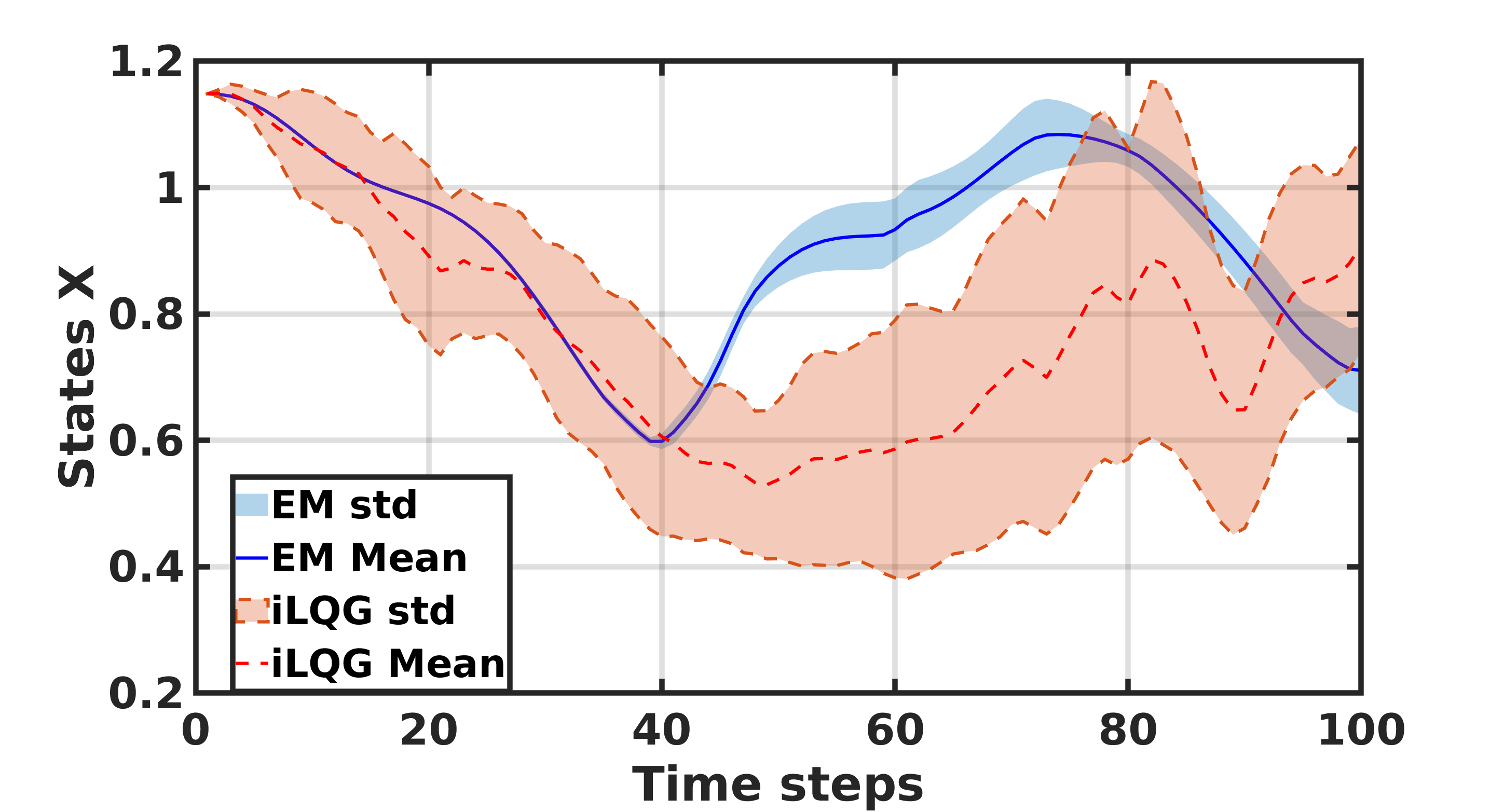}%
		\end{subfigure} 
		\begin{subfigure}{\figsmallwidth}
			\includegraphics[width=\columnwidth]{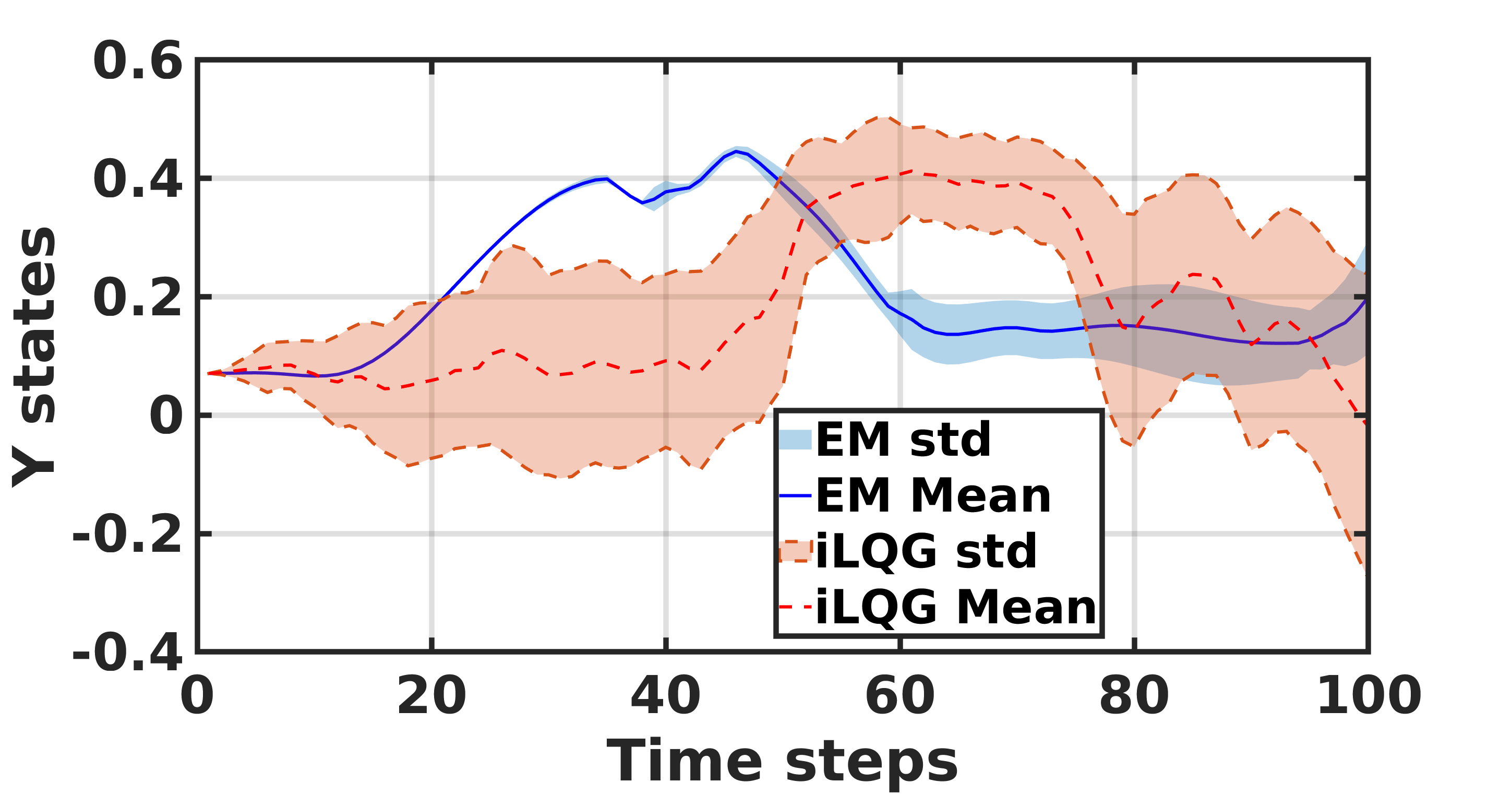}%
		\end{subfigure} 
		\begin{subfigure}{\figsmallwidth}
			\includegraphics[width=\columnwidth]{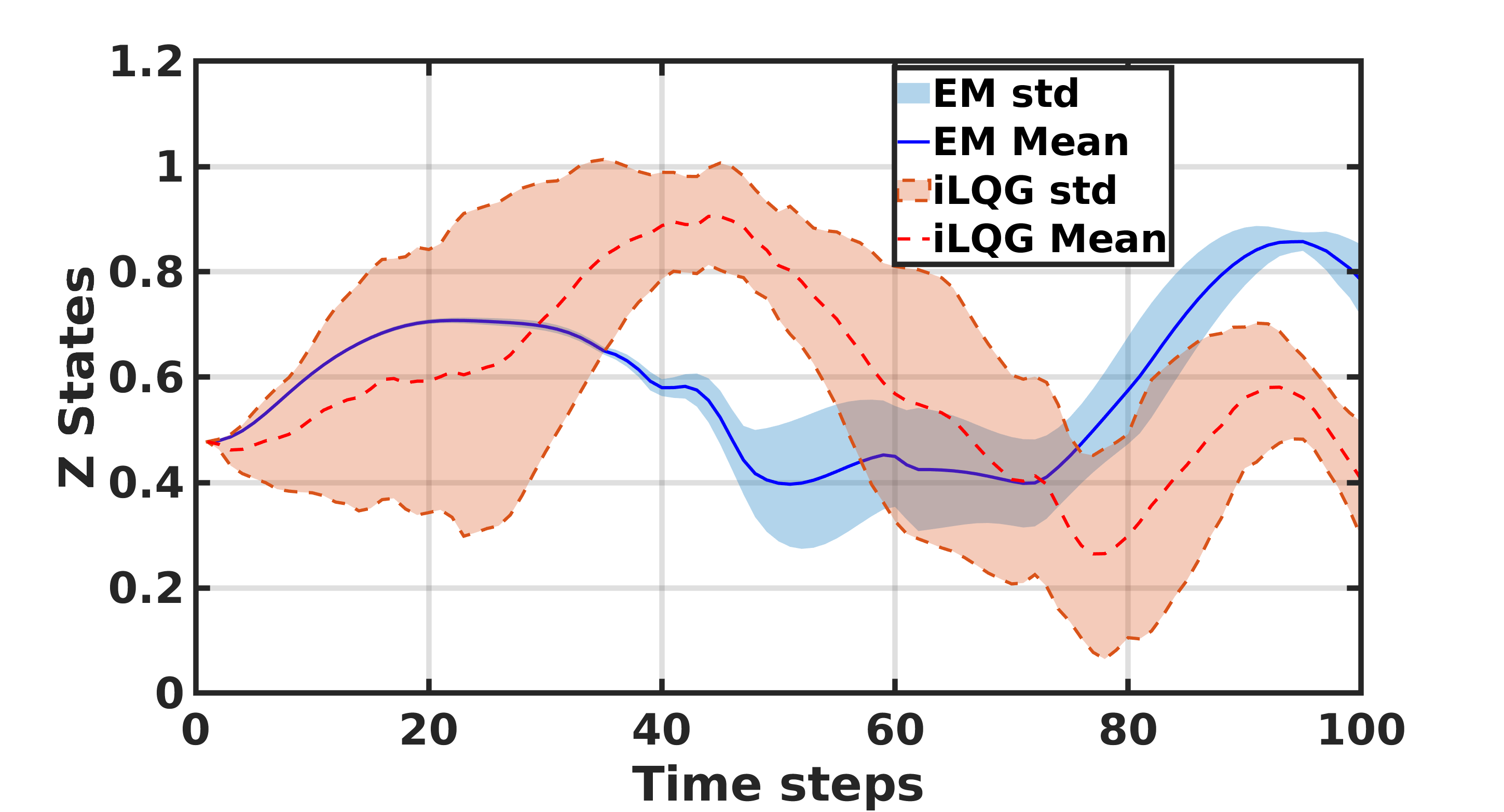}%
		\end{subfigure}
		\caption{Evolution of three states corresponding to position $[x,y,z]$ of the door-opening task, where the positions represent that of the hook which is connected to the end-effector. Each plot of this figure is divided into two coloured plots that shows the mean and variance of the state trajectories. The dashed lines represent $\phi$ obtained from iLQG and solid lines represent trajectories obtained from $\phi^*$ of eq~\eqref{proposition:lastEq}.}.
		\label{fig:contributionStateNoiseDoorOpening}
	\end{figure*}
\section{Experimental Results} \label{sec:results3}
We evaluated our method on simulations conducted on a platform known as \texttt{MUJOCO}; please see \cite{mujoco} for details of the kinematics and dynamics of the models embedded in the framework which utilizes a system of rigid links with noisy motors at the joints. Inside the \texttt{MUJOCO} framework we used a model of a \texttt{PR2} robot to try to address two tasks as mentioned below:
\begin{enumerate}
	\item Door opening task: This task requires opening a door with a $6$ degree of freedom (DOF) $3D$ robotic arm. In this task the arm must grasp the handle and then pull the door to a target angle, which is challenging for model-based methods due to the complicated contact dynamics between the hand and the handle. A contact needs to be established before the door can be opened. The cost function utilized is a weighted combination of the distance of the end effector to the door handle with that of the angle of the door. The right sub-figure of Fig~\ref{fig:contribution3MujocoSnaps} shows an instance of the simulation.
	\item Peg in a hole task: This particular task requires inserting the peg which is attached to the end effector of the \texttt{PR2} robot in a hole which is created in a board placed infront of \texttt{PR2} manipulator. The left sub-figure of Fig~\ref{fig:contribution3MujocoSnaps} depicts an instance of the task of inserting the peg in the hole.
\end{enumerate}
We consider sensor noise $\mbepsilon^s_k$ to be Gaussian in our experiments, i.e.,
\begin{align}
	\mathbf{s}_k &= \mathbf{x}_k+ \mbepsilon^s_k ,\;   \mbepsilon^s_k \sim \mathcal{N} (0,  \rho^2 \mathbf{I}_{n_s} ), 
	\label{cov_a_k}
\end{align}
where $\mathbf{x}_k$ is the real state. 
The sensor noise $\mbepsilon^s_k$  is indirectly propagated into the design of control action $\mathbf{a}_k$ that  forms the real  input 
to the system. For conducting the experiments I utilized $\rho^2=0.1$ and followed the steps as mentioned in
Algorithm~\ref{alg:DPSOCEM}.

{\it \textbf{Cost:}} 
For the door opening task, the mean$\pm$std-dev of the expected real  costs
$\sum_{l=1}^T Y_l(\mathbf{s}_l,\mathbf{a}_l)$ for $14$ iterations are depicted in the bottom two sub-plots of Fig.~\ref{fig:contribution3CostILQGEMILQG}. The datasets take into account $20$ samples; the bottom left utilizes Algorithm~\ref{alg:DPSOCEM} whereas the bottom right is when a particular baseline, i.e., iLQG, is used for $14$ continuous iterations.
The noise parameter was set as $\rho^2=0.1$. The continuous line in the bottom sub-plots in blue represents the mean of the cost for a particular iteration and the box plots represent a variation of the costs for iterations. The bottom right plot provides expected costs for a particular baseline i.e., iLQG for $14$ iterations and the bottom left plot shows the course of expected costs if on fourth iteration EM-based parameters $\hat{\phi}^{i+1}$ is utilized and then the iLQG is repeated from the $6^{th}$ iterations. A very interesting result can be seen as compared to running full iLQG iterations for $14$ iterations, without the intervention of EM-based optimization. The top plot shows the average min angle in radians of the door hinge  to which the door is pulled by the hook attached to the end effector. It can be clearly seen that the EM samples of the door provide a better minimum angle of opening compared to the baseline.

For the peg in the hole task, the mean$\pm$std-dev of the expected real costs
$\sum_{l=1}^T Y_l(\mathbf{s}_l,\mathbf{a}_l)$ for $13$ iterations for this task is shown in~\ref{fig:contributionStatesPeg}. I utilized $\phi^{12}$ as an intialization for obtaining $\hat{\phi}^{i+1}$ based on~\eqref{proposition:lastEq}

{\it \textbf{Trajectories:}} We compare the true state trajectories, i.e., $ \{\textbf{x}_1,..,\textbf{x}_T\}$ produced on the real platform excited 
by the control actions with the parameters obtained through repeated EM iterations. In this experiment, iLQG was used 
as the baseline. We performed 20 experiments and each experiment ran for 14 subsequent iterations. 
The mean$\pm$std-dev of the trajectories $[x_k, y_k, z_k]$, $k=1,\cdots,T$, for the position of the end-effector is illustrated in Fig.~\ref{fig:contribution3StateNoise} and Fig.~\ref{fig:contributionStateNoiseDoorOpening} for the tasks of peg in the hole and door opening, respectively.
It  can be observed that evolution of the state trajectories 
demonstrates better performance achieved by EM iterations in terms of less stochasticity in the states. In particular, there is a high variance in the dashed trajectories compared to the solid lines. The advantage gained by the EM iterations can be explained by the less deviation caused by noise. 
It has been shown that the exploitation functionality of the EM approach 
can effectively reduce the stochasticity in control policies. Fig.~\ref{fig:contributionStateNoiseDoorOpening} shows the evolution of the end effector position, i.e., $[x,y,z]$ states of the hook. All plots are divided into solid and dashed lines and the respective deviation is from a set of 20 experiments conducted from the same initial condition with a total horizon of the experiment being $T=100$. The superiority of the DPSOC-EM is noticeable in terms of the uncertainty reduction by the generated optimal control estimates. A similar profile, can be seen resulting from Fig~\ref{fig:contribution3StateNoise}, where each plot exhibits two variables corresponding to the top and bottom end of the peg; e.g., the top plot of Fig~\ref{fig:contribution3StateNoise} shows the $x$ position of the two ends of the peg, where the solid lines and respective deviations (in shade) correspond to the EM parameter estimates and the dash lines and the deviations correspond to the iLQG parameter estimates. This figure show the temporal evolution of states for a total time $T=70$.

\textbf{\textit{Performance comparison in the video:}} 
For the two tasks, the simulation video for the comparison of iLQG and EM-based optimal control can be visualized via gif file that can be found in the following \href{https://github.com/proxymallick/EM-Guided-Policy-Search}{Github} link.
In the above shown video links, the left video corresponds to the baseline policy and the right one corresponds to the parameters derived as a result of the approach mentioned in this chapter.
It can be visualised in the simulation (for both tasks) that the stochasticity of end effector is evident from the jittery nature of the joints of the \texttt{PR2} robot when its trying to accomplish the two tasks.

{
	{\it \textbf{Exploitation efficiency in control actions:}} 
	Fig~\ref{fig:contribution3ControlNoise} shows $6$ plots each corresponding to the control actions of joints of the PR$2$ robot for the door opening task as well as the peg in the hole tasks depicted in Fig~\ref{fig:contribution3MujocoSnaps}. Referring to Fig~\ref{fig:contribution3ControlNoise}, one can see that each sub-plot contains $4$ box plots with two different color. Each color shade is specific to each task, i.e., blue corresponds to Peg in hole task and red corresponds to door opening task. Each box plot shows a variation of actions, $a_k^{n} \quad \forall k  \in \{1,\cdots, T\} \quad \text{and} \quad \forall {n} \in \{1,\cdots, 40\},$ where $n$ denotes samples/trajectories.
Each color has two pairs of results, $\phi^i$ and $\hat{\phi}^{i+1}$ denoting the baseline iLQG and EM parameters respectively. One can notice that the superiority in the proposed optimization approach is evident from the lesser deviations in-terms of control action values. The observed improvement in control actions 
	is attributed to the exploitation mechanism in the DPSOC-EM approach
	that significantly reduces the stochasticity (variation) in the control policies over multiple samples over the the total horizon.}

The simulations for Step 1, 2 and 5 were performed using the 64-bit \textsf{Ubuntu} 16.04 OS on  \textsf{Dell Alienware} 15 R2 
of  \textsf{Intel} Core i7-6700HQ CPU @ 2.60GHz.
The simulations for Step 3 and 4 were conducted using multiple 2.6 GHz \textsf{Intel Xeon Broadwell (E5-2697A v4)} processors on the high performance computing (HPC) grid located at The University of Newcastle.

\section{Conclusion}
This paper focuses on two major parts; a) provides an alternate version of approximate DDP i.e., DPSOC-EM, b) delivers novel connections between information theory and stochastic control dualities with optimisation of mixture likelihood. The new approximate control paradigm has a similar analytical architecture as that mentioned in the paper~\cite[Section 4]{mallick2022stochastic}. As a result of this structure, the analytical guarantees of surrogate function holds. It has also been shown empirically to generate superior estimates as compared to baselines, e.g., the superiority of the DPSOC-EM is noticeable in terms of the uncertainty reduction (both in the control action and the states) by the generated estimates and also the reduction of the overall cost function.

\appendix

\subsection{Some EM Lemmas}

 \begin{lemma} \label{lemma:EM_proof_new}
Consider $L_\phi ({\mathbb{Y}_T})$ and $\mathcal{L}(\phi,\hat{\phi}^{i})$ defined 
in \eqref{Lphi} and \eqref{Ltheta_kheta}, respectively, with 
$\hat{\phi^i}$ a  known parameter estimate.
Let \begin{align}
     l(\phi, \widetilde{p}(\mathbb{S}_{T+1}) ) =\mathbb{E}_{ \widetilde{p}(\mathbb{S}_{T+1})} 
    \log  \frac{  {p_\phi (\mathbb{S}_{T+1},\mathbb{Y}_T)} }{  \widetilde{p}(\mathbb{S}_{T+1})} \label{lphi}
\end{align}
for any distribution  $\widetilde{p}(\mathbb{S}_{T+1})$.
One has 
\begin{align*}
  L_\phi ({\mathbb{Y}_T}) \geq  l(\phi, \widetilde{p}(\mathbb{S}_{T+1}) ), 
\end{align*}
that is, $ l(\phi, \widetilde{p}(\mathbb{S}_{T+1}) ) $ is a lower bound of $ {  L_\phi ({\mathbb{Y}_T})} $.
Moreover,  let \begin{align} \label{eq:166}
\widetilde{p}(\mathbb{S}_{T+1}) = p_{ \hat{\phi}^i} (\mathbb{S}_{T+1}|\mathbb{Y}_T)  ,
\end{align}
and  denote
\begin{align}
 l(\phi,  \hat{\phi}^i) = l(\phi, p_{\hat{\phi}^i} (\mathbb{S}_{T+1}|\mathbb{Y}_T) ). \label{lphiphii}
 \end{align}
One has
\begin{align}
\hat{ \phi}^{i*}  =  \arg\max_{\phi}  l(\phi,  \hat{\phi}^i) = \arg\max_{\phi} \mathcal{L} (\phi, \hat\phi^i). \label{maxlmaxcalL}
\end{align}
\end{lemma}

\begin{proof} 
One can start by expanding $L_\phi ({\mathbb{Y}_T})$ as follows 
\begin{align}
    L_\phi(\mathbb{Y}_T)  & =  \log  p_\phi(\mathbb{Y}_T)  \nonumber\\
    &=  \log  \int p_\phi(\mathbb{Y}_T,\mathbb{S}_{T+1}) d\mathbb{S}_{T+1} \nonumber\\
    & = \log  \int \widetilde{p} (\mathbb{S}_{T+1}) \frac{p_\phi (\mathbb{S}_{T+1},\mathbb{Y}_T)}{\widetilde{p}(\mathbb{S}_{T+1})} d\mathbb{S}_{T+1}  \nonumber\\
    &  \geq \int \widetilde{p}(\mathbb{S}_{T+1}) \log  \frac{p_\phi (\mathbb{S}_{T+1},\mathbb{Y}_T)}{\widetilde{p}(\mathbb{S}_{T+1})} d\mathbb{S}_{T+1}  \nonumber \\
    &  =\mathbb{E}_{ \widetilde{p}(\mathbb{S}_{T+1})} 
     \log  \frac { {p_\phi (\mathbb{S}_{T+1},\mathbb{Y}_T)} }{  \widetilde{p}(\mathbb{S}_{T+1})} \nonumber\\
      &  = l(\phi, \widetilde{p}(\mathbb{S}_{T+1})),\label{loerbound_eq}
\end{align}
which holds for any distribution $\widetilde{p}(\mathbb{S}_{T+1})$. Nevertheless, one requires $\widetilde{p}(\mathbb{S}_{T+1})$ to deliver an optimal bound such that, at a known value (current estimate) of $\phi$, the bound $l(\phi, \widetilde{p}(\mathbb{S}_{T+1}))$ touches $ L_\phi(\mathbb{Y}_T)$.  We can further expand the lower bound $l(\phi, \widetilde{p}(\mathbb{S}_{T+1}))$ as follows,
\begin{align}
   & l(\phi, \widetilde{p}(\mathbb{S}_{T+1})) \nonumber \\
     = & \int  \widetilde{p}(\mathbb{S}_{T+1}) \log   \frac{p_\phi(\mathbb{S}_{T+1}, \mathbb{Y}_{T})}{ \widetilde{p}(\mathbb{S}_{T+1})} d\mathbb{S}_{T+1} \nonumber \\
     = & \int  \widetilde{p}(\mathbb{S}_{T+1}) \log   \frac{p_\phi(\mathbb{S}_{T+1}| \mathbb{Y}_{T}) p_\phi(\mathbb{Y}_T) }{ \widetilde{p}(\mathbb{S}_{T+1})} d\mathbb{S}_{T+1} \nonumber \\
     =  & \int  \widetilde{p}(\mathbb{S}_{T+1}) \log   \frac{p_\phi(\mathbb{S}_{T+1}| \mathbb{Y}_{T})  }{ \widetilde{p}(\mathbb{S}_{T+1})} d\mathbb{S}_{T+1} \nonumber  \\
       &+ \int  \widetilde{p}(\mathbb{S}_{T+1}) \log   {p_\phi( \mathbb{Y}_{T})  } d\mathbb{S}_{T+1}  \nonumber \\
        = &  \mathbb{E}_{ \widetilde{p}(\mathbb{S}_{T+1})} \log   \frac{p_\phi(\mathbb{S}_{T+1}| \mathbb{Y}_{T})  }{ \widetilde{p}(\mathbb{S}_{T+1})} + \mathbb{E}_{\widetilde{p}(\mathbb{S}_{T+1})} {  \log  p_\phi (\mathbb{Y}_T)} \nonumber  \\
       =&  - \text{D}_{\text{KL}} (\widetilde{p}(\mathbb{S}_{T+1} ) || p_\phi (\mathbb{S}_{T+1}|\mathbb{Y}_T) + \log  p_\phi (\mathbb{Y}_T) \label{Lemma_1_property} . \end{align}
It indicates that the tight optimal distribution \eqref{eq:166}  is  clearly is a consequence of $\text{D}_{\text{KL}} (\cdot) =0$ 
as $\text{D}_{\text{KL}} \geq 0$. So, the  selection \eqref{eq:166} basically utilizes the known estimate $\hat{\phi}^i$ to solve for $\widetilde{p}(\mathbb{S}_{T+1})$. 
This step of solving for optimal distribution of $\widetilde{p}(\mathbb{S}_{T+1})$ is called the E-step.  

With \eqref{eq:166},  we can denote \eqref{lphiphii}, which is re-calculated below, 
noting \eqref{Ltheta_kheta} and \eqref{lphi},
\begin{align}
    l(\phi,\hat{\phi}^i) =&  \mathbb{E}_{ {\hat{\phi}^i}}  (\log  {p_\phi (\mathbb{S}_{T+1},\mathbb{Y}_T)}|\mathbb{Y}_T)
     \nonumber\\
    & - \mathbb{E}_{p_{\hat{\phi}^i} ( \mathbb{S}_{T+1}|\mathbb{Y}_T)}  (\log  p_{\hat{\phi}^i} (\mathbb{S}_{T+1}|\mathbb{Y}_T))      \nonumber\\
 =& \mathcal{L} (\phi, \hat{\phi}^i) - \mathbb{E}_{p_{\hat{\phi}^i} ( \mathbb{S}_{T+1}|\mathbb{Y}_T)}  (\log  p_{\hat{\phi}^i} (\mathbb{S}_{T+1}|\mathbb{Y}_T)).  \label{eq:28_new}
\end{align}
Maximizing $l(\phi,\hat{\phi}^i)$ with respect to $\phi$ is exactly equivalent to maximizing $\mathcal{L} (\phi,\hat{\phi}^i)$ as the second term on the RHS of \eqref{eq:28_new} is independent on $\phi$, that is, 
\eqref{maxlmaxcalL} is proved.   In particular, the subsequent step for 
maximizing $\mathcal{L} (\phi, \hat\phi^i)$ is called the M-step.
\end{proof}
\begin{lemma}
  \label{lemma:EM_proof}
Suppose the parameter $\hat{\phi}^{i+1}$ is produced in an iteration, 
one that 
\begin{align}
L_{\hat{\phi}^{i+1}} (\mathbb{Y}_T)    -  L_{\hat{\phi}^i} (\mathbb{Y}_T) 
\geq \mathcal{L} (\hat{\phi}^{i+1},\hat{\phi}^i) -  \mathcal{L} (\hat{\phi}^i,\hat{\phi}^i) \label{LLcalLcalL}
\end{align}
where the equality holds if 
$p_{\hat{\phi}^{i+1}} (\mathbb{S}_{T+1}|\mathbb{Y}_T) =  p_{\hat\phi^{i}} (\mathbb{S}_{T+1}|\mathbb{Y}_T)$ . 
\end{lemma}
\begin{proof}
The proof is the fundamental development of \cite{dempster1977maximum} and other details of assumptions and convergence can be found in \cite{mclachlan2007algorithm}, \cite{wu1983convergence}. One can expand  $\mathcal{L} (\phi,\hat{\phi}^i)$ in a straightforward manner, using \eqref{Ltheta_kheta},
\begin{align*}
     \mathcal{L} (\phi,\hat{\phi}^i) =&  \mathbb{E}_{\hat{\phi^i}} (  \log  p_\phi (\mathbb{S}_{T+1}, \mathbb{Y}_T) | \mathbb{Y}_T ) 
     \\
      =& \mathbb{E}_{p_{\hat{\phi}^i} (\mathbb{S}_{T+1}|\mathbb{Y}_T)} (\log  p_{{\phi} } (\mathbb{S}_{T+1}|\mathbb{Y}_T)) \nonumber \\
     & + \mathbb{E}_{p_{\hat{\phi}^i} (\mathbb{S}_{T+1}|\mathbb{Y}_T)} (\log  p_\phi (\mathbb{Y}_T) ) \\
    =& \mathbb{E}_{p_{\hat{\phi}^i} (\mathbb{S}_{T+1}|\mathbb{Y}_T)} (\log  p_{{\phi} } (\mathbb{S}_{T+1}|\mathbb{Y}_T))  +  \log  p_\phi (\mathbb{Y}_T).
\end{align*}
Therefore,
\begin{align*}
     & \mathcal{L} (\phi,\hat{\phi}^i) -  \mathcal{L} (\hat{\phi}^i,\hat{\phi}^i) \\
=& \mathbb{E}_{p_{\hat{\phi}^i} (\mathbb{S}_{T+1}|\mathbb{Y}_T)} \log 
\frac{ p_{{\phi} } (\mathbb{S}_{T+1}|\mathbb{Y}_T) }{p_{\hat{\phi}^i } (\mathbb{S}_{T+1}|\mathbb{Y}_T)} 
+  L_\phi (\mathbb{Y}_T)    -  L_{\hat{\phi}^i} (\mathbb{Y}_T)\\
=& -\text{D}_{\text{KL}} ({p}_{\hat{\phi}^i} (\mathbb{S}_{T+1}|\mathbb{Y}_T) || {p}_{\phi} (\mathbb{S}_{T+1}|\mathbb{Y}_T))  +  L_\phi (\mathbb{Y}_T)    -  L_{\hat{\phi}^i} (\mathbb{Y}_T) \\
\leq & L_\phi (\mathbb{Y}_T)    -  L_{\hat{\phi}^i} (\mathbb{Y}_T),
\end{align*}
which verifies \eqref{LLcalLcalL} with $\phi=\hat{\phi}^{i+1}$.
\end{proof}

\bibliographystyle{IEEEtran}
\bibliography{IEEEabrv,main} 

\begin{thebibliography}{10}
\providecommand{\url}[1]{#1}
\csname url@samestyle\endcsname
\providecommand{\newblock}{\relax}
\providecommand{\bibinfo}[2]{#2}
\providecommand{\BIBentrySTDinterwordspacing}{\spaceskip=0pt\relax}
\providecommand{\BIBentryALTinterwordstretchfactor}{4}
\providecommand{\BIBentryALTinterwordspacing}{\spaceskip=\fontdimen2\font plus
\BIBentryALTinterwordstretchfactor\fontdimen3\font minus
  \fontdimen4\font\relax}
\providecommand{\BIBforeignlanguage}[2]{{%
\expandafter\ifx\csname l@#1\endcsname\relax
\typeout{** WARNING: IEEEtran.bst: No hyphenation pattern has been}%
\typeout{** loaded for the language `#1'. Using the pattern for}%
\typeout{** the default language instead.}%
\else
\language=\csname l@#1\endcsname
\fi
#2}}
\providecommand{\BIBdecl}{\relax}
\BIBdecl

\bibitem{todorov2005generalized}
E.~Todorov and W.~Li, ``A generalized iterative lqg method for locally-optimal
  feedback control of constrained nonlinear stochastic systems,'' in
  \emph{Proceedings of the 2005, American Control Conference}.\hskip 1em plus
  0.5em minus 0.4em\relax IEEE, 2005, pp. 300--306.

\bibitem{toussaint2009robot}
M.~Toussaint, ``Robot trajectory optimization using approximate inference,'' in
  \emph{Proceedings of the 26th annual international conference on machine
  learning}.\hskip 1em plus 0.5em minus 0.4em\relax ACM, 2009, pp. 1049--1056.

\bibitem{Schulman13findinglocally}
J.~Schulman, J.~Ho, A.~Lee, I.~Awwal, H.~Bradlow, and P.~Abbeel, ``Finding
  locally optimal, collision-free trajectories with sequential convex
  optimization,'' in \emph{in Proc. Robotics: Science and Systems}, 2013.

\bibitem{mallick2022stochastic}
P.~Mallick and Z.~Chen, ``Stochastic optimal control for multivariable
  dynamical systems using expectation maximization,'' \emph{IEEE Transactions
  on Neural Networks and Learning Systems}, 2022.

\bibitem{mallick2022reinforcement}
P.~Mallick, Z.~Chen, and M.~Zamani, ``Reinforcement learning using expectation
  maximization based guided policy search for stochastic dynamics,''
  \emph{Neurocomputing}, vol. 484, pp. 79--88, 2022.

\bibitem{mujoco}
E.~Todorov, T.~Erez, and Y.~Tassa, ``Mujoco: A physics engine for model-based
  control,'' in \emph{Intelligent Robots and Systems (IROS), 2012 IEEE/RSJ
  International Conference on}.\hskip 1em plus 0.5em minus 0.4em\relax IEEE,
  2012, pp. 5026--5033.

\bibitem{chebotar2017combining}
Y.~Chebotar, K.~Hausman, M.~Zhang, G.~Sukhatme, S.~Schaal, and S.~Levine,
  ``Combining model-based and model-free updates for trajectory-centric
  reinforcement learning,'' in \emph{Proceedings of the 34th International
  Conference on Machine Learning-Volume 70}.\hskip 1em plus 0.5em minus
  0.4em\relax JMLR. org, 2017, pp. 703--711.

\bibitem{theodorouinformation}
E.~A. Theodorou and E.~Todorov, ``Relative entropy and free energy dualities:
  Connections to path integral and kl control,'' in \emph{51st IEEE Conference
  on Decision and Control (CDC)}.\hskip 1em plus 0.5em minus 0.4em\relax IEEE,
  2012, pp. 1466--1473.

\bibitem{theodorou2015nonlinear}
E.~A. Theodorou, ``Nonlinear stochastic control and information theoretic
  dualities: Connections, interdependencies and thermodynamic
  interpretations,'' \emph{Entropy}, vol.~17, no.~5, pp. 3352--3375, 2015.

\bibitem{dai1996connections}
P.~Dai~Pra, L.~Meneghini, and W.~J. Runggaldier, ``Connections between
  stochastic control and dynamic games,'' \emph{Mathematics of Control, Signals
  and Systems}, vol.~9, no.~4, pp. 303--326, 1996.

\bibitem{Billingsley1995probability}
P.~Billingsley, \emph{Probability and measure}.\hskip 1em plus 0.5em minus
  0.4em\relax John Wiley \& Sons, 2008.

\bibitem{stengel1994optimal}
R.~F. Stengel, \emph{Optimal Control and Estimation}.\hskip 1em plus 0.5em
  minus 0.4em\relax Courier Corporation, 1994.

\bibitem{10.2307/3689975}
\BIBentryALTinterwordspacing
C.~H. Papadimitriou and J.~N. Tsitsiklis, ``The complexity of markov decision
  processes,'' \emph{Mathematics of Operations Research}, vol.~12, no.~3, pp.
  441--450, 1987. [Online]. Available:
  \url{http://www.jstor.org/stable/3689975}
\BIBentrySTDinterwordspacing

\bibitem{platt2010belief}
R.~Platt~Jr, R.~Tedrake, L.~Kaelbling, and T.~Lozano-Perez, ``Belief space
  planning assuming maximum likelihood observations,'' 2010.

\bibitem{li2004iterative}
W.~Li and E.~Todorov, ``Iterative linear quadratic regulator design for
  nonlinear biological movement systems.'' in \emph{ICINCO (1)}, 2004, pp.
  222--229.

\bibitem{tassa2012synthesis}
Y.~Tassa, T.~Erez, and E.~Todorov, ``Synthesis and stabilization of complex
  behaviors through online trajectory optimization,'' in \emph{Intelligent
  Robots and Systems (IROS), 2012 IEEE/RSJ International Conference on}.\hskip
  1em plus 0.5em minus 0.4em\relax IEEE, 2012, pp. 4906--4913.

\bibitem{levine2014motor}
S.~Levine, ``Motor skill learning with local trajectory methods,'' Ph.D.
  dissertation, Stanford University, 2014.

\bibitem{ziebart2010modeling(a)}
B.~D. Ziebart, J.~A. Bagnell, and A.~K. Dey, ``Modeling interaction via the
  principle of maximum causal entropy,'' in \emph{International Conference on
  Machine Learning}, 2010.

\bibitem{kappen2012optimal}
H.~J. Kappen, V.~G{\'o}mez, and M.~Opper, ``Optimal control as a graphical
  model inference problem,'' \emph{Machine Learning}, vol.~87, no.~2, pp.
  159--182, 2012.

\bibitem{todorov2007linearly}
E.~Todorov, ``Linearly-solvable markov decision problems,'' in \emph{Advances
  in Neural Information Processing Systems}, 2007, pp. 1369--1376.

\bibitem{littman1994memoryless}
M.~L. Littman, ``Memoryless policies: Theoretical limitations and practical
  results,'' in \emph{From Animals to Animats 3: Proceedings of the third
  international conference on simulation of adaptive behavior}, vol.~3.\hskip
  1em plus 0.5em minus 0.4em\relax Cambridge, MA, 1994, p. 238.

\bibitem{lusena2001nonapproximability}
C.~Lusena, J.~Goldsmith, and M.~Mundhenk, ``Nonapproximability results for
  partially observable markov decision processes,'' \emph{Journal of Artificial
  Intelligence Research}, vol.~14, pp. 83--103, 2001.

\bibitem{levine2016end}
S.~Levine, C.~Finn, T.~Darrell, and P.~Abbeel, ``End-to-end training of deep
  visuomotor policies,'' \emph{The Journal of Machine Learning Research},
  vol.~17, no.~1, pp. 1334--1373, 2016.

\bibitem{bishop2006pattern}
C.~M. Bishop, \emph{Pattern Recognition and Machine Learning}.\hskip 1em plus
  0.5em minus 0.4em\relax springer, 2006.

\bibitem{cooper2013method}
G.~F. Cooper, ``A method for using belief networks as influence diagrams,''
  \emph{arXiv preprint arXiv:1304.2346}, 2013.

\bibitem{dayan1997using}
P.~Dayan and G.~E. Hinton, ``Using expectation-maximization for reinforcement
  learning,'' \emph{Neural Computation}, vol.~9, no.~2, pp. 271--278, 1997.

\bibitem{norouzi2016reward}
M.~Norouzi, S.~Bengio, N.~Jaitly, M.~Schuster, Y.~Wu, D.~Schuurmans
  \emph{et~al.}, ``Reward augmented maximum likelihood for neural structured
  prediction,'' in \emph{Advances In Neural Information Processing Systems},
  2016, pp. 1723--1731.

\bibitem{levine2014learning}
S.~Levine and P.~Abbeel, ``Learning neural network policies with guided policy
  search under unknown dynamics,'' in \emph{Advances in Neural Information
  Processing Systems}, 2014, pp. 1071--1079.

\bibitem{bertsekas1995dynamic}
D.~P. Bertsekas, \emph{Dynamic Programming and Optimal Control}.\hskip 1em plus
  0.5em minus 0.4em\relax Athena scientific Belmont, MA, 1995, vol.~1, no.~2.

\bibitem{chebotar2017path}
Y.~Chebotar, M.~Kalakrishnan, A.~Yahya, A.~Li, S.~Schaal, and S.~Levine, ``Path
  integral guided policy search,'' in \emph{2017 IEEE international conference
  on robotics and automation (ICRA)}.\hskip 1em plus 0.5em minus 0.4em\relax
  IEEE, 2017, pp. 3381--3388.

\bibitem{dempster1977maximum}
A.~P. Dempster, N.~M. Laird, and D.~B. Rubin, ``Maximum likelihood from
  incomplete data via the em algorithm,'' \emph{Journal of the Royal
  Statistical Society. Series B}, pp. 1--38, 1977.

\bibitem{mclachlan2007algorithm}
G.~McLachlan and T.~Krishnan, \emph{The EM Algorithm and Extensions}.\hskip 1em
  plus 0.5em minus 0.4em\relax John Wiley \& Sons, 2007, vol. 382.

\bibitem{wu1983convergence}
C.~J. Wu \emph{et~al.}, ``On the convergence properties of the em algorithm,''
  \emph{The Annals of Statistics}, vol.~11, no.~1, pp. 95--103, 1983.

\end{thebibliography}

\begin{IEEEbiography}[{\includegraphics[width=1in,height=1.25in,clip,keepaspectratio]{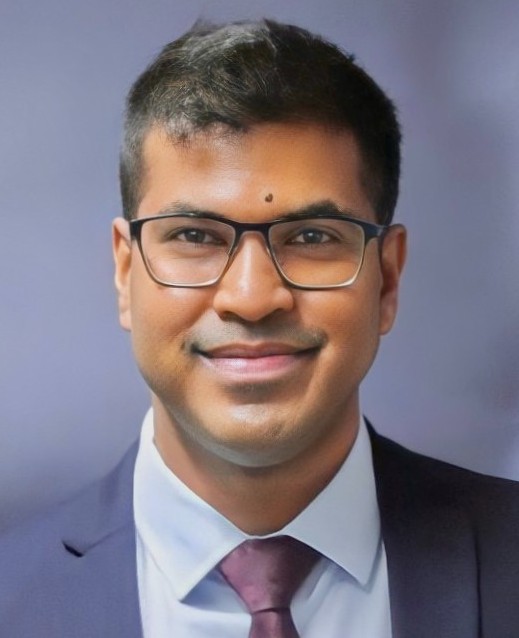}}]
{Prakash Mallick} holds a B.Tech, M.Eng, and Ph.D. from the National Institute of Technology, Rourkela, India, University of Melbourne, Australia, and University of Newcastle, Australia, respectively (awarded in 2012, 2017, and 2023). With extensive industrial research and development experience, including three years as an electrical engineer in Coal India Limited and one year in quantitative research at Ardea Investment Management, he currently serves as a postdoctoral research fellow at the Australian Institute of Machine Learning (AIML) funded by Defence Science and Technology Group (DSTG), Adelaide. \end{IEEEbiography}

\begin{IEEEbiography}[{\includegraphics[width=1.1in,height=1.25in,clip,keepaspectratio]{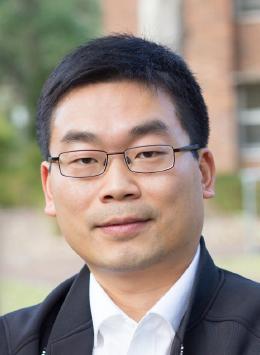}}]
{Zhiyong Chen} received the B.E. degree from the University of Science and Technology of China, and the M.Phil. and Ph.D. degrees from the Chinese University of Hong Kong, in 2000, 2002 and 2005, respectively. He worked as a Research Associate at the University of Virginia during 2005-2006. He joined the University of Newcastle, Australia, in 2006, where he is currently a Professor. He was also a Changjiang Chair Professor with Central South University, Changsha, China. His research interests include non-linear systems and control, biological systems, and multi-agent systems. He is/was an associate editor of Automatica, IEEE Transactions on Automatic Control and IEEE Transactions on Cybernetics.  \end{IEEEbiography}

\vfill

\end{document}